\documentclass{article}
\usepackage{xypic}
\usepackage{tikz}
\usetikzlibrary{calc}
\usetikzlibrary{arrows}
\usepackage[bookmarksopen=true]{hyperref}

\input{alternation.sty}

\title{A finite alternation result for reversible boolean circuits}

\author{Peter Selinger}

\date{Dalhousie University}

\begin{document}
\maketitle

\begin{abstract}
  We say that a reversible boolean function on $n$ bits has {\em
    alternation depth} $\dee$ if it can be written as the sequential
  composition of $\dee$ reversible boolean functions, each of which
  acts only on the top $n-1$ bits or on the bottom $n-1$
  bits. Moreover, if the functions on $n-1$ bits are even, we speak of
  {\em even alternation depth}. We show that every even reversible
  boolean function of $n\geq 4$ bits has alternation depth at most 9
  and even alternation depth at most 13.
\end{abstract}

\section{Introduction}

A reversible boolean function on $n$ bits is a permutation of
$\s{0,1}^n$.  It is well-known that the NOT, controlled NOT, and
Toffoli gates form a universal gate set for reversible boolean
functions {\cite{Toffoli1980,Musset1997,DVRS2002}}. More precisely,
these gates generate (via the operations of composition and cartesian
product, and together with the identity functions) all reversible
boolean functions on $n$ bits, when $n\leq 3$, and all even reversible
boolean functions on $n$ bits, when $n\geq 4$. A particular
representation of a reversible boolean function in terms of these
generators is called a {\em reversible circuit}. The problem of
finding a (preferably short) circuit to implement a given reversible
function is called the {\em synthesis problem} {\cite{SM2013}}.

When working with reversible boolean functions and circuits, it is not
typically possible to reason inductively; we cannot usually reduce a
problem about circuits on $n$ bits to a problem about circuits on
$n-1$ bits. In this paper, we prove a theorem that may, in some cases,
make such inductive reasoning possible: we prove that when $n\geq 4$,
every even reversible function on $n$ bits can be decomposed into at
most 9 reversible functions on $n-1$ bits:
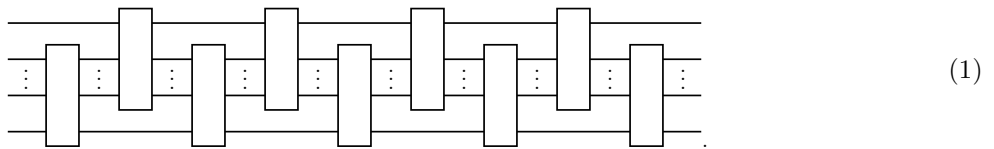
\begin{equation}\label{eqn-circuit}
  \m{\begin{qcircuit}[scale=0.6]
      \grid{19}{0,1,2,3};
      \vdotslabel{0.5,1};
      \biggate{}{1.5,0}{1.5,2};
      \vdotslabel{2.5,1};
      \biggate{}{3.5,1}{3.5,3};
      \vdotslabel{4.5,1};
      \biggate{}{5.5,0}{5.5,2};
      \vdotslabel{6.5,1};
      \biggate{}{7.5,1}{7.5,3};
      \vdotslabel{8.5,1};
      \biggate{}{9.5,0}{9.5,2};
      \vdotslabel{10.5,1};
      \biggate{}{11.5,1}{11.5,3};
      \vdotslabel{12.5,1};
      \biggate{}{13.5,0}{13.5,2};
      \vdotslabel{14.5,1};
      \biggate{}{15.5,1}{15.5,3};
      \vdotslabel{16.5,1};
      \biggate{}{17.5,0}{17.5,2};
      \vdotslabel{18.5,1};
    \end{qcircuit}.}
\end{equation}
If, moreover, each of the functions on $n-1$ bits is also required to
be even, we prove that a decomposition into 13 such functions is
possible.  It is of course not remarkable that $n$-bit circuits can be
decomposed into $(n-1)$-bit circuits: after all, we already know that
they can be decomposed into $3$-bit circuits, namely gates.  What is
perhaps remarkable is that the bound 9 (respectively, 13) on the depth
is independent of $n$.

There are some potential applications of such a result --- although
admittedly, they may not be very practical. As a first application,
one may obtain an alternative proof of universality, by turning any
universal gate set on $n$ bits into a universal gate set on $n+1$
bits, provided that $n\geq 3$. This also yields a new method for
circuit synthesis: given a good procedure for synthesizing even
$n$-bit circuits, we obtain a procedure for synthesizing even
$(n+1)$-bit circuits that is at most $13$ times worse. By applying
this idea recursively, we obtain circuits of size $O(13^n)$ for any
reversible function on $n$ bits. This is worse than what can be
obtained by other methods. However, it may be possible to improve this
procedure further, for example by noting that the 13 subcircuits need
not be completely general; they can be chosen to be of particular
forms, which may be easier to synthesize recursively.

Another potential application is the presentation of (even) reversible
boolean functions by generators and relations. While the NOT, CNOT,
and Toffoli gates are a well-known set of generators, to the author's
knowledge, no complete set of relations for these generators is
known. For any given $n$, the group of $n$-bit reversible functions is
a finite group, so finding a complete set of relations for any fixed
$n$ is a finite (although very large) problem. However, it is not
trivial to find a set of relations that works for all $n$; at present,
it is not even known whether the theory is finitely axiomatizable. If
we had a procedure for rewriting every circuit into one of the form
{\eqref{eqn-circuit}}, then we could obtain a complete set of
relations for $n$-bit circuits by considering (a) a complete set of
relations for $(n-1)$-bit circuits, (b) the relations required to do
the rewriting, and (c) any relations required to prove equalities
between circuits of the form {\eqref{eqn-circuit}}. In particular, if
it could be shown that a finite set of relations is sufficient for
(b) and (c), a finite equational presentation of reversible boolean
functions could be derived.

Finally, the task of realizing a given permutation with low
alternation depth can also make for an entertaining puzzle. Such a
puzzle has been implemented and is available from {\cite{puzzle}}.

\section{Statement of the main result}

We write $\S(X)$ for the group of permutations of a finite set $X$.
For $f\in\S(X)$ and $g\in\S(Y)$, let $f\times g\in\S(X\times Y)$ be
the permutation defined componentwise by
$(f\times g)(x,y) = (f(x),g(y))$.  We also write $\id_X\in\S(X)$ for
the identity permutation on $X$.  Recall that a permutation is {\em
  even} if it can be written as a product of an even number of
2-cycles.

Let $2=\{0,1\}$ be the set of booleans, which we identify with the
binary digits $0$ and $1$. By abuse of notation, we also write
$2 = \id_2$ for the identity permutation on the set $2$.

\begin{definition}
  Let $A$ be a finite set, and let $\sigma\in\S(2\times A\times 2)$ be
  a permutation. We say that $\sigma$ has {\em alternation depth} $\dee$
  if it can be written as a product of $\dee$ factors
  $\sigma=\sigma_1\sigma_2\cdots\sigma_\dee$, where each factor
  $\sigma_i$ is either of the form $f\times 2$ for some
  $f\in\S(2\times A)$ or of the form $2\times g$ for some
  $g\in\S(A\times 2)$.
\end{definition}

The purpose of this paper is to prove the following theorem:

\begin{theorem}\label{thm-main}
  Let $A$ be a finite set of 3 or more elements. Then every even
  permutation $\sigma \in \S(2\times A\times 2)$ has alternation depth
  9.
\end{theorem}

In circuit notation, Theorem~\ref{thm-main} can be understood as
stating that every reversible boolean function on the set
$2\times A\times 2$ can be expressed as a circuit in the following
form:
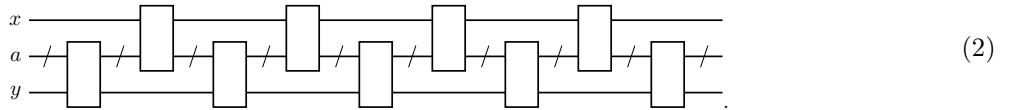
\begin{equation}\label{eqn-circuit2}
  \m{\begin{qcircuit}[scale=0.6]
      \grid{19}{0,1,2};
      \leftlabel{$x$}{0,2};
      \leftlabel{$a$}{0,1};
      \leftlabel{$y$}{0,0};
      \multiwire{0.5,1};
      \biggate{}{1.5,0}{1.5,1};
      \multiwire{2.5,1};
      \biggate{}{3.5,1}{3.5,2};
      \multiwire{4.5,1};
      \biggate{}{5.5,0}{5.5,1};
      \multiwire{6.5,1};
      \biggate{}{7.5,1}{7.5,2};
      \multiwire{8.5,1};
      \biggate{}{9.5,0}{9.5,1};
      \multiwire{10.5,1};
      \biggate{}{11.5,1}{11.5,2};
      \multiwire{12.5,1};
      \biggate{}{13.5,0}{13.5,1};
      \multiwire{14.5,1};
      \biggate{}{15.5,1}{15.5,2};
      \multiwire{16.5,1};
      \biggate{}{17.5,0}{17.5,1};
      \multiwire{18.5,1};
    \end{qcircuit}.
  }
\end{equation}
Here, the lines labelled $x$ and $y$ each represent a bit, and the
line labelled $a$ represents an element of the set $A$. The case of
boolean circuits arises as the special case where the cardinality of
$A$ is a power of $2$.

\begin{remark}\label{rem-even}
  The evenness of $\sigma$ is a necessary condition for
  Theorem~\ref{thm-main}, because all permutations of the forms
  $f\times 2$ and $2\times g$ are even, and therefore only even
  permutations can have an alternation depth.
\end{remark}

\begin{remark}\label{rem-even-main}
  Our definition of alternation depth does not require that the
  permutations $f\in\S(2\times A)$ and $g\in\S(A\times 2)$ are
  themselves even. However, if Theorem~\ref{thm-main} is to be applied
  recursively (as required, for example, by the potential applications
  mentioned in the introduction), we need each $f$ and $g$ to be
  even. Since the proof of Theorem~\ref{thm-main} is already
  complicated enough without this restriction, we do not consider the
  case of even $f$ and $g$ until Section~\ref{sec-even}, where we
  prove that the alternation depth is at most 13 in that case.
\end{remark}

Our proof of Theorem~\ref{thm-main} is in two parts. In
Section~\ref{sec-first}, we will show that every even permutation of a
certain form $g+h$ has alternation depth 5.  In
Section~\ref{sec-second}, we will show that every even permutation can
be decomposed into a permutation of alternation depth 4 and a
permutation of the form $g+h$. Together, these results imply
Theorem~\ref{thm-main}.

\section{First construction: balanced permutations}\label{sec-first}

\subsection{Preliminaries}\label{ssec-prelim}

We fix some terminology. The {\em support} of a permutation
$\sigma\in\S(X)$ is the set
$\supp\sigma = \s{x\in X\mid \sigma(x)\neq x}$.  Two permutations
$\sigma,\tau\in\S(X)$ are {\em disjoint} if
$\supp\sigma\cap\supp\tau = \emptyset$. In this case, $\sigma$ and
$\tau$ commute: $\sigma\tau = \tau\sigma$. We also call $\sigma\tau$ a
{\em disjoint product} in this case.  Recall the cycle notation for
permutations: for $\kay>1$, we write $(a_1\;a_2\;\ldots\;a_\kay)$, or
sometimes $(a_1,a_2,\ldots,a_\kay)$, for the permutation with support
$\s{a_1,\ldots,a_\kay}$ defined by $a_1\mapsto a_2$, $a_2\mapsto a_3$,
\ldots, $a_{\kay-1}\mapsto a_\kay$ and $a_\kay\mapsto a_1$. Such a
permutation is also called a {\em $\kay$-cycle}.  Every permutation
can be uniquely decomposed (up to the order of the factors) into a
product of disjoint cycles.  A $\kay$-cycle is even if and only if
$\kay$ is odd.

Two permutations $\sigma,\sigma'\in\S(X)$ are {\em similar}, in symbols
$\sigma\sim\sigma'$, if there exists $\tau$ such that
$\sigma'=\tau\inv\sigma\tau$. It is easy to see that $\sigma,\sigma'$
are similar if and only if their cycle decompositions contain an equal
number of $\kay$-cycles for every $\kay$.

If $g,h\in\S(X)$ are permutations on some finite set $X$, we define
their {\em disjoint sum} $g+h\in\S(2\times X)$ by
\[ (g+h)(0,x) = (0,g(x)) \quad\mbox{and}\quad 
(g+h)(1,x) = (1,h(x)).
\]
We note the following properties:
\begin{eqnarray}
  g+g &=& 2\times g, \label{eqn-1} \\
  (g+h)\times 2 &=& g\times 2 + h\times 2, \label{eqn-2} \\
  (g+h)(g'+h') &=& gg' + hh'. \label{eqn-3}
\end{eqnarray}
Property {\eqref{eqn-1}} also helps explain our choice of
writing ``$2$'' for the identity permutation in $\S(2)$.

\begin{figure}
  \[
  \begin{tikzpicture}
    \begin{scope}
      \path (-0.35, 2.2) node {(a)};
      \begin{scope}[rounded corners=2pt,fill=blue!10,draw=blue!80]
        \filldraw (0,0) rectangle +(1.3,1);
        \path (0.65,0.5) node[color=blue!80] {$A$};
        \filldraw (1.5,0) rectangle +(1.3,1);
        \path (2.15,0.5) node[color=blue!80] {$A$};
        \filldraw (0,1.2) rectangle +(1.3,1);
        \path (0.65,1.7) node[color=blue!80] {$A$};
        \filldraw (1.5,1.2) rectangle +(1.3,1);
        \path (2.15,1.7) node[color=blue!80] {$A$};
      \end{scope}
      \begin{scope}[shift={(0.65,1.1)}]
        \draw[->,xscale=0.4] (-65:1cm) arc (-65:245:1cm);
        \path[xscale=0.4] (140:1cm) node[left=-1pt] {$f$};
      \end{scope}
      \begin{scope}[shift={(2.15,1.1)}]
        \draw[->,xscale=0.4] (-65:1cm) arc (-65:245:1cm);
        \path[xscale=0.4] (140:1cm) node[left=-1pt] {$f$};
      \end{scope}
      \path (1.4,-0.4) node {$f\times 2$};
    \end{scope}
    \begin{scope}[shift={(4,0)}]
      \path (-0.35, 2.2) node {(b)};
      \begin{scope}[rounded corners=2pt,fill=blue!10,draw=blue!80]
        \filldraw (0,0) rectangle +(1.3,1);
        \path (0.65,0.5) node[color=blue!80] {$A$};
        \filldraw (1.5,0) rectangle +(1.3,1);
        \path (2.15,0.5) node[color=blue!80] {$A$};
        \filldraw (0,1.2) rectangle +(1.3,1);
        \path (0.65,1.7) node[color=blue!80] {$A$};
        \filldraw (1.5,1.2) rectangle +(1.3,1);
        \path (2.15,1.7) node[color=blue!80] {$A$};
      \end{scope}
      \begin{scope}[shift={(1.4,0.5)}]
        \draw[->,yscale=0.4,xscale=1.1] (-65:1cm) arc (-65:245:1cm);
        \path[yscale=0.4,xscale=1.1] (160:1cm) node[left=1pt] {$g$};
      \end{scope}
      \begin{scope}[shift={(1.4,1.7)}]
        \draw[->,yscale=0.4,xscale=1.1] (-65:1cm) arc (-65:245:1cm);
        \path[yscale=0.4,xscale=1.1] (160:1cm) node[left=1pt] {$g$};
      \end{scope}
      \path (1.4,-0.4) node {$2\times g$};
    \end{scope}
    \begin{scope}[shift={(8,0)}]
      \path (-0.35, 2.2) node {(c)};
      \begin{scope}[rounded corners=2pt,fill=blue!10,draw=blue!80]
        \filldraw (0,0) rectangle +(1.3,1);
        \path (0.65,0.5) node[color=blue!80] {$A$};
        \filldraw (1.5,0) rectangle +(1.3,1);
        \path (2.15,0.5) node[color=blue!80] {$A$};
        \filldraw (0,1.2) rectangle +(1.3,1);
        \path (0.65,1.7) node[color=blue!80] {$A$};
        \filldraw (1.5,1.2) rectangle +(1.3,1);
        \path (2.15,1.7) node[color=blue!80] {$A$};
      \end{scope}
      \begin{scope}[shift={(1.4,0.5)}]
        \draw[->,yscale=0.4,xscale=1.1] (-65:1cm) arc (-65:245:1cm);
        \path[yscale=0.4,xscale=1.1] (160:1cm) node[left=1pt] {$g$};
      \end{scope}
      \begin{scope}[shift={(1.4,1.7)}]
        \draw[<-,yscale=0.4,xscale=1.1] (25:1cm) arc (25:335:1cm);
        \path[yscale=0.4,xscale=1.1] (160:1cm) node[left=1pt] {$h$};
      \end{scope}
      \path (1.4,-0.4) node {$g+h$};
    \end{scope}
  \end{tikzpicture}
  \]
  \caption{Visualizing permutations of $2\times A\times 2$}\label{fig-square}
\end{figure}
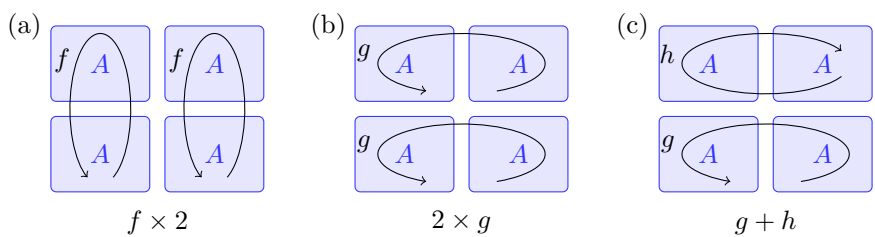

Although it will not be strictly necessary for the proofs that follow
(which are combinatorial), it may sometimes be helpful to visualize
sets of the form $2\times A\times 2$, and permutations thereon, as
follows.  We visualize the set $2\times A$ as two copies of $A$
stacked vertically, with elements of the form $(0,a)$ and $(1,a)$
belonging to the lower and upper copy, respectively. Similarly, we
visualize the set $A\times 2$ as two copies of $A$ side by side, with
elements of the form $(a,0)$ and $(a,1)$ belonging to the left and
right copy, respectively. In the same vein, we visualize the set
$2\times A\times 2$ as four copies of $A$ arranged in two rows and
columns. The effect of a permutations of the form $f\times 2$ is to
apply $f$ separately to the left and right column, as shown in
Figure~\ref{fig-square}(a). Similarly, the effect of $2\times g$ is to
apply $g$ separately to the top and bottom rows, and the effect of
$g+h$ is to apply $g$ to the bottom row and $h$ to the top row, as
shown in Figure~\ref{fig-square}(b) and~(c).

\subsection{Decomposition into balanced permutations}

\begin{definition}
  A permutation $\sigma$ is {\em balanced} if the number of $\kay$-cycles
  in its cycle decomposition is even for all $\kay\geq 2$. Moreover,
  $\sigma$ is {\em nearly balanced} if the number of $\kay$-cycles in its
  cycle decomposition is even for all $\kay\geq 3$.
\end{definition}

For example, the permutation $(1\;2)(3\;4)(5\;6\;7)(8\;9\;10)$ is
balanced, the permutation $(1\;2)(3\;4\;5)(6\;7\;8)$ is nearly
balanced, and $(1\;2)(3\;4)(5\;6\;7)$ is neither balanced nor nearly
balanced.

\begin{remark}\label{rem-balanced}
  The disjoint product of any number of (nearly) balanced permutations
  is (nearly) balanced. Moreover, a nearly balanced permutation is
  balanced if and only if it is even.
\end{remark}

The purpose of this subsection is to prove that every even permutation on
a set of 5 or more elements can be decomposed into a product of two
balanced permutations. This will be Proposition~\ref{pro-balanced} below.
The proof requires a sequence of lemmas.

\begin{lemma}\label{lem-decomposition-1}
  Let $\sigma$ be a $\kay$-cycle, where $\kay\geq 2$ and $\kay\neq 3$. Then
  there exists a balanced permutation $\rho$ and a nearly balanced
  permutation $\tau$ such that $\sigma=\tau\rho$. Moreover,
  $\supp\tau\cup\supp\rho\seq\supp\sigma$.
\end{lemma}

\begin{proof}
  Let $\sigma=(a_{1}\; a_{2}\; \ldots\; a_{\kay})$. If $\kay=2t$ is even,
  let
  \[ \begin{array}{r@{~}c@{~}l}
    \rho &=& (a_{1}\; a_{2}\; \ldots\; a_{t})(a_{t+1}\; a_{t+2}\;
    \ldots\; a_{2t}),
    \\
    \tau &=& (a_{1}\;a_{t+1}).
  \end{array}
  \]
  If $\kay=2t+1$ is odd (and therefore, by assumption, $t\geq 2$), let
  \[ \begin{array}{r@{~}c@{~}l}
    \rho &=& (a_{1}\; a_{2}\; \ldots\; a_{t})(a_{t+1}\; a_{t+3}\;
    a_{t+4}\; \ldots\; a_{2t+1}),
    \\
    \tau &=& (a_{1}\;a_{t+1})(a_{t+2}\; a_{t+3}).
  \end{array}
  \]
  In both cases, the conclusion of the lemma is satisfied.
\end{proof}

\begin{lemma}\label{lem-decomposition-2}
  Let $\sigma$ be the disjoint product of a 3-cycle and a $\kay$-cycle,
  where $\kay\geq 2$. Then there exists a balanced permutation $\rho$ and
  a nearly balanced permutation $\tau$ such that
  $\sigma=\tau\rho$. Moreover, $\supp\tau\cup\supp\rho\seq\supp\sigma$.
\end{lemma}

\begin{proof}
  Let $\sigma=(b_{1}\;b_{2}\;b_{3})(a_{1}\;a_{2}\;\ldots\;a_{\kay})$.
  If $\kay=2$, let
  \[ \begin{array}{r@{~}c@{~}l}
    \rho &=& (b_{1}\;b_{2})(a_{1}\;a_{2}),
    \\
    \tau &=& (b_{1}\;b_{3}).
  \end{array}
  \]
  If $\kay=3$, let
  \[ \begin{array}{r@{~}c@{~}l}
    \rho &=& (b_{1}\;b_{2})(a_{1}\;a_{2}),
    \\
    \tau &=& (b_{1}\;b_{3})(a_{1}\;a_{3}),
  \end{array}
  \]
  If $\kay=4$, let
  \[ \begin{array}{r@{~}c@{~}l}
    \rho &=& (b_{1}\;b_{2}\;b_{3})(a_{1}\;a_{2}\;a_{3}),
    \\
    \tau &=& (a_{1}\;a_{4}).
  \end{array}
  \]
  If $\kay=2t+1$ is odd, with $t\geq 2$, let
  \[ \begin{array}{r@{~}c@{~}l}
    \rho &=& (b_{1}\;b_{2}\;b_{3})(a_{t}\; a_{t+1}\; a_{2t+1}),
    \\
    \tau &=& (a_{1}\;a_{2}\;\ldots\;a_{t})(a_{t+2}\;a_{t+3}\;\ldots\;a_{2t+1}).
  \end{array}
  \]
  If $\kay=2t$ is even, with $t\geq 3$, let
  \[ \begin{array}{r@{~}c@{~}l}
    \rho &=&
    (a_{3}\;a_{4}\;\ldots\;a_{t}\;a_{2t})(a_{t+1}\;a_{t+2}\;\ldots\;a_{2t-1}),
    \\
    \tau &=& (b_{1}\;b_{2}\;b_{3})(a_{1}\;a_{2}\;a_{3})(a_{t+1}\;a_{2t}).
  \end{array}
  \]
  In all cases, the conclusion of the lemma is satisfied.
\end{proof}

\begin{lemma}\label{lem-decomposition-3}
  Let $\sigma$ be a disjoint product of two or more 3-cycles. Then
  there exist balanced permutations $\rho,\tau$ such that
  $\sigma=\tau\rho$. Moreover, $\supp\tau\cup\supp\rho\seq\supp\sigma$.
\end{lemma}

\begin{proof}
  By assumption, $\sigma$ can be factored as
  $\sigma=\gamma_1\gamma_2\cdots\gamma_\ell$, where
  $\gamma_1,\ldots,\gamma_\ell$ are pairwise disjoint 3-cycles and
  $\ell\geq 2$. Note that $\gamma_i^2$ is also a 3-cycle, and
  $\gamma_i^4=\gamma_i$, for all $i$.

  If $\ell$ is even, let
  $\rho=\tau=\gamma_1^2\gamma_2^2\cdots\gamma_{\ell}^2$.  If $\ell$ is
  odd, let $\rho=\gamma_1\gamma_2^2\gamma_3^2\cdots\gamma_{\ell-1}^2$
  and $\tau=\gamma_2^2\gamma_3^2\cdots\gamma_{\ell-1}^2\gamma_{\ell}$.
  In both cases, the conclusion of the lemma is satisfied.
\end{proof}

\begin{lemma}\label{lem-decomposition-3b}
  Let $\sigma$ be an even permutation, other than a 3-cycle. Then
  $\sigma$ can be written as $\sigma=\tau\rho$, where $\rho,\tau$ are
  balanced.
\end{lemma}

\begin{proof}
  By considering the cycle decomposition of $\sigma$, it is easy to
  see that $\sigma$ can be factored into disjoint factors such that
  each factor satisfies the premise of one of
  Lemmas~\ref{lem-decomposition-1}, {\ref{lem-decomposition-2}}, or
  {\ref{lem-decomposition-3}}. Let $\sigma=\sigma_1\cdots\sigma_\ell$ be
  such a factorization. Using the lemmas, each $\sigma_i$ can be
  written as $\sigma_i=\tau_i\rho_i$, where $\rho_i$ is balanced and
  $\tau_i$ is nearly balanced. Moreover, since the support of each
  $\rho_i$ and $\tau_i$ is contained in that of $\sigma_i$, the
  $\rho_i$ are pairwise disjoint, the $\tau_i$ are pairwise disjoint,
  and $\rho_i\tau_j = \tau_j\rho_i$ whenever $i\neq j$. Let
  $\rho = \rho_1\cdots\rho_\ell$ and $\tau=\tau_1\cdots\tau_\ell$. Then we
  have $\sigma=\tau\rho$. Moreover, by Remark~\ref{rem-balanced},
  $\rho$ is balanced and $\tau$ is nearly balanced. Finally, since
  $\sigma$ and $\rho$ are even permutations, so is $\tau$, and it
  follows, again by Remark~\ref{rem-balanced}, that $\tau$ is
  balanced.
\end{proof}

\begin{lemma}\label{lem-decomposition-4}
  Let $\sigma$ be a 3-cycle in $\S(X)$, where $\card{X}\geq 5$.  Then there
  exist balanced permutations $\rho,\tau$ such that $\sigma=\tau\rho$.
\end{lemma}

\begin{proof}
  Let $\sigma=(a_{1}\;a_{2}\;a_{3})$. Since $\card{X}\geq 5$, there exist
  elements $a_{4},a_{5}$ of $X$ that are different from
  each other and from $a_{1},\ldots,a_{3}$. Let
  \[ \begin{array}{r@{~}c@{~}l}
    \rho &=& (a_{1}\;a_{2})(a_{4}\;a_{5}),
    \\
    \tau &=& (a_{1}\;a_{3})(a_{4}\;a_{5}).
  \end{array}
  \]
  Then the conclusion of the lemma is satisfied.
\end{proof}

\begin{remark}
  Unlike the situation in
  Lemmas~\ref{lem-decomposition-1}--\ref{lem-decomposition-3b}, it is
  not possible to choose $\rho$ and $\tau$ in
  Lemma~\ref{lem-decomposition-4} so that their support is contained
  in that of $\sigma$. An easy case distinction shows that
  Lemma~\ref{lem-decomposition-4} is false when $\card{X}\leq 4$. 
\end{remark}

\begin{proposition}\label{pro-balanced}
  Let $\sigma$ be an even permutation in $\S(X)$, where $\card{X}\geq 5$.
  Then there exist balanced permutations $\rho,\tau$ such that
  $\sigma=\tau\rho$.
\end{proposition}

\begin{proof}
  By Lemma~\ref{lem-decomposition-4} if $\sigma$ is a 3-cycle, and by
  Lemma~\ref{lem-decomposition-3b} otherwise.
\end{proof}

\subsection{Alternation depth of permutations of the form \texorpdfstring{$\id+\tau$}{id+t}}

Using balanced permutations, we prove that every even permutation of
the form $\id+\tau$ has alternation depth 5.

\begin{proposition}\label{pro-construction-1+h}
  Let $A$ be a finite set of 3 or more elements, and let
  $\tau\in\S(A\times 2)$ be an even permutation. Let
  $\sigma = \id_{A\times 2}+\tau \in \S(2\times A\times 2)$. Then
  $\sigma$ has an alternating decomposition of depth 5, starting and
  ending with a factor of the form $2\times g$.
\end{proposition}

The proof requires two lemmas.

\begin{lemma}\label{lem-balanced}
  Let $\tau\in\S(A\times 2)$ be a balanced permutation. Then there
  exist permutations $g\in\S(A\times 2)$ and $h\in\S(A)$ such that
  \[ \tau = g\inv(h\times 2)g.
  \]
\end{lemma}

\begin{proof}
  For all $\kay\geq 2$, let $\yy{\kay}$ be the number of $\kay$-cycles
  in the cycle decomposition of $\tau$. Since the cycles are disjoint,
  we have $\sum_\kay \kay \yy{\kay} \leq 2\card{A}$. Since $\tau$ is
  balanced, all $\yy{\kay}$ are even. We can therefore find a
  permutation $h\in\S(A)$ whose number of $\kay$-cycles is exactly
  $\yy{\kay}/2$, for all $\kay$. Since $h\times 2$ and $\tau$ have, by
  construction, the same number of $\kay$-cycles for all $\kay$, we
  have $h\times 2\sim\tau$.  By definition of similarity, it follows
  that there exists some $g$ with $\tau = g\inv(h\times 2)g$, as
  claimed.
\end{proof}

\begin{lemma}\label{lem-balanced2}
  Let $\tau\in\S(A\times 2)$ be a balanced permutation, and let
  $\sigma=\id_{A\times 2}+\tau \in \S(2\times A\times 2)$. Then there
  exist permutations $g\in\S(A\times 2)$ and $f\in\S(2\times A)$ such
  that
  \[ \sigma = (2\times g\inv)(f\times 2)(2\times g).
  \]
\end{lemma}

\begin{proof}
  By Lemma~\ref{lem-balanced}, we can find $g\in\S(A\times 2)$ and
  $h\in\S(A)$ such that $\tau = g\inv(h\times 2)g$. Let
  $f=\id_A+h \in \S(2\times A)$.  Then
  \[ 
  \begin{array}{rcl}
    \multicolumn{3}{l}{(2\times g\inv)(f\times 2)(2\times g)} \\
    &=& (2\times g\inv)((\id_A+h)\times 2)(2\times g) \\
    &=& (g\inv + g\inv)(\id_A\times 2+h\times 2)(g + g) \\
    &=& (g\inv\, \id_{A\times 2}\, g) + (g\inv(h\times 2)g) \\
    &=& \id_{A\times 2} + \tau \\
    &=& \sigma.
  \end{array}
  \]
  Here, in addition to the defining properties of $f$, $g$, $h$, and
  $\sigma$, we have also used {\eqref{eqn-1}} and {\eqref{eqn-2}} in
  the second step and {\eqref{eqn-3}} in the third step. 
\end{proof}

\begin{proof}[Proof of Proposition~\ref{pro-construction-1+h}]
  By assumption, $\tau\in\S(A\times 2)$ is even. By
  Proposition~\ref{pro-balanced}, there exist balanced permutations
  $\tau_1,\tau_2\in\S(A\times 2)$ such that $\tau=\tau_2\tau_1$. By
  Lemma~\ref{lem-balanced2}, there exist $g_1,g_2\in\S(A\times 2)$ and
  $f_1,f_2\in\S(2\times A)$ such that
  $\id+\tau_i = (2\times g_i\inv)(f_i\times 2)(2\times
  g_i)$, for $i=1,2$. Then we have:
  \[ 
  \begin{array}{rcl}
    \id+\tau &=& (\id+\tau_2)(\id+\tau_1) \\
           &=& (2\times g_2\inv)(f_2\times 2)(2\times g_2)
               (2\times g_1\inv)(f_1\times 2)(2\times g_1) \\
           &=& (2\times g_2\inv)(f_2\times 2)(2\times g_2g_1\inv)
               (f_1\times 2)(2\times g_1),
  \end{array}
  \]
  which is the desired decomposition of alternation depth 5.
\end{proof}

\subsection{Alternation depth of permutations of the form $g+h$}

We now come to the main result of Section~\ref{sec-first}, which is
that every even permutation of the form $g+h$ has alternation depth 5.

\begin{proposition}\label{pro-construction-1}
  Let $A$ be a finite set of 3 or more elements, and let
  $g,h\in\S(A\times 2)$ be permutations such that $\sigma=g+h$ is
  even. Then $\sigma$ has alternation depth 5.
\end{proposition}

\begin{proof}
  This is an easy corollary of Proposition~\ref{pro-construction-1+h}.
  Let $\tau=hg\inv\in\S(A\times 2)$, and note that $\tau$ is even.  By
  Proposition~\ref{pro-construction-1+h}, $\id_{A\times 2}+\tau$ can
  be written in the form $\id+\tau = (2\times g_5)(f_4\times 2)
  (2\times g_3)(f_2\times 2)(2\times g_1)$.  We then have
  \[
  \begin{array}{rcl}
  \sigma &=& g+h \\
    &=& (\id+\tau)(g+g) \\
    &=& (\id+\tau)(2\times g) \\
    &=& (2\times g_5)(f_4\times 2)(2\times g_3)(f_2\times 2)(2\times g_1g),
  \end{array}
  \]
  which proves that $\sigma$ has alternation depth 5.
\end{proof}

\section{Second construction: colorings}\label{sec-second}

\subsection{Colorings}

As before, let $2=\s{0,1}$. If $X$ is any finite set, a {\em coloring}
of $X$ is a map $c : X \to 2$. Here, we think of the binary digits $0$
and $1$ as {\em colors}, i.e., $x\in X$ has color $c(x)$. We say that
the coloring $c$ is {\em fair} if there is an equal number of elements
of each color, i.e., $\card{c\inv\s{0}} = \card{c\inv\s{1}}$.

The group $\S(X)$ acts in a natural way on the colorings of $X$ as
follows: we define $\sigma\bu c = c'$, where
$c'(x) = c(\sigma\inv(x))$.  Note that
$(\sigma\tau)\bu c = \sigma \bu (\tau\bu c)$. Also, $\sigma\bu c$ is
fair if and only if $c$ is fair.

On a set of the form $2\times X$, the {\em standard coloring} is the
one given by $\cstan(0,x) = 0$ and $\cstan(1,x) = 1$, for all $x$.

\begin{remark}\label{rem-fair}
  The standard coloring is fair. Conversely, if $c$ is a fair coloring
  of $2\times X$, there exists a permutation $f\in\S(2\times X)$ such
  that $f\bu c = \cstan$.
\end{remark}

The following lemma relates colorings to permutations of the form
$g+h$ considered in the previous section.

\begin{lemma}\label{lem-glue}
  A permutation $\sigma\in\S(2\times X)$ is of the form $\sigma=g+h$,
  for some $g,h\in\S(X)$, if and only if $\sigma\bu\cstan=\cstan$.
\end{lemma}

\begin{proof}
  This is elementary. We have $\sigma\bu\cstan=\cstan$ if and only if
  for all $x$, $\sigma(0,x)$ is of the form $(0,y)$, and $\sigma(1,x)$
  is of the form $(1,z)$. By setting $g(x)=y$ and $h(x)=z$, this is
  equivalent to $\sigma$ being of the form $g+h$.
\end{proof}

We are now ready to state the main result of Section~\ref{sec-second},
which is that every fair coloring of $2\times A\times 2$ can be
converted to the standard coloring by the action of a permutation of
alternation depth 4.

\begin{proposition}\label{pro-construction-2}
  Let $A$ be a finite set of 3 or more elements, and let $c$ be a fair
  coloring of $2\times A\times 2$. Then there exists a permutation
  $\sigma\in\S(2\times A\times 2)$ such that $\sigma\bu c = \cstan$ and
  $\sigma$ has alternation depth 4.
\end{proposition}

The proof of Proposition~\ref{pro-construction-2} will take up the
remainder of Section~\ref{sec-second}.

\subsection{Visualizing colorings}\label{ssec-vis-col}

Colorings on $2\times A\times 2$ can be visualized in the same
row-and-column format we used in Figure~\ref{fig-square}. An example
of a coloring, where $A=\s{p,q,r}$, is shown in
Figure~\ref{fig-vis-col}(a). The figure indicates, for example, that
$c(1,p,0)=0$, $c(1,q,0)=1$, and so on. When the names of the elements
of $A$ are not important, we omit them. Additionally, we sometimes
represent the colors 0 and 1 by black and white squares, respectively,
as in Figure~\ref{fig-vis-col}(b).

\begin{figure}
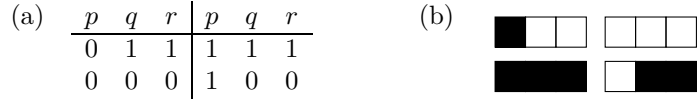

  \[
  \text{(a)}\quad
  \begin{array}[t]{ccc|ccc}
    p & q & r & p & q & r \\
    \hline
    0 & 1 & 1 & 1 & 1 & 1 \\
    0 & 0 & 0 & 1 & 0 & 0 \\
  \end{array}
  \qquad\qquad
  \text{(b)}\quad
  \begin{array}[t]{c}
    \\[-1.5ex]
    \lgsquares{0.4}{011}~~\lgsquares{0.4}{111}\\[.5mm]
    \lgsquares{0.4}{000}~~\lgsquares{0.4}{100}
  \end{array}
  \]
  \caption{Visualizing colorings of $2\times A\times 2$}\label{fig-vis-col}
\end{figure}

\subsection{Color pairs}\label{ssec-colorpair}

We begin by characterizing when two colorings $c,c'$ of $2\times X$
are related by the action of a permutation of the form $2\times g$ for
$g\in\S(X)$. This is the case if and only if $c$ and $c'$ have the
same {\em color pair distribution}.

\begin{definition}
  Let $X$ be a set, and consider a coloring $c$ of $2\times X$. We
  define a function $\dd{c} : X \to 2\times 2$ by
  $\dd{c}(x) = (c(0,x), c(1,x))$.  We call $\dd{c}(x)$ the {\em color
    pair} of $x$.
\end{definition}

Informally, a color pair corresponds to a single column of digits in
Figure~\ref{fig-vis-col}(a).  We note that the action of permutations
$g\in\S(X)$ respects color pairs in the following sense: let
$c' = (2\times g)\bu c$. Then
\begin{equation}\label{eqn-d-c}
  \dd{c'}(g(x)) = (c'(0,g(x)), c'(1,g(x))) = (c(0,x), c(1,x)) =
  \dd{c}(x).
\end{equation}
In particular, the action of $2\times g$ on colorings does not change
the {\em number} of elements of $X$ with each color pair. Conversely,
whenever two colorings $c,c'$ have this property, then they are
related by the action of $2\times g$, for some $g$. The following
definition helps us state this more precisely.

\begin{definition}
  Let $X$ be a set, and $c$ a coloring of $2\times X$. For any
  $i,j\in 2$, define $N_c(i,j)\seq X$ to be the set of elements with
  color pair $(i,j)$, i.e.,
  \[ N_c(i,j) = \s{x\in X\mid \dd{c}(x) = (i,j)}.
  \]
  Note that $X$ is the disjoint union of the $N_c(i,j)$, for $i,j\in2$.
  Let $n_c(i,j) = \card{N_c(i,j)}$ be the number of elements with
  color pair $(i,j)$. Then the {\em color pair distribution} of $c$ is
  the 4-tuple
  \[ (n_c(0,0), n_c(0,1), n_c(1,0), n_c(1,1)).
  \]
\end{definition}

For example, the coloring from Figure~\ref{fig-vis-col} has color pair
distribution $(1,4,0,1)$, because the color pair $(0,0)$ occurs once,
the color pair $(0,1)$ occurs four times, and so on. The following
lemma is then obvious.

\begin{lemma}\label{lem-color-pair-dist}
  Let $c,c'$ be colorings of $2\times X$. Then $c,c'$ have the same
  color pair distribution if and only if there exists a permutation
  $g\in\S(X)$ such that $c' = (2\times g)\bu c$. \qed
\end{lemma}

\subsection{Color standardization}

\begin{definition}
  Let $A$ be a set, and let $c$ be a coloring of $2\times A\times 2$.
  We say that $c$ is
  \begin{itemize}
  \item {\em standard} if $c=\cstan$, i.e., if $\dd{c}(a,0)=\dd{c}(a,1)=(0,1)$
    for all $a\in A$;
  \item {\em symmetric} if $\dd{c}(a,0) = \dd{c}(a,1)$ for all $a\in A$;
  \item {\em regular} if each color pair occurs an even number of
    times, i.e., if $n_c(0,0)$, $n_c(0,1)$, $n_c(1,0)$, and $n_c(1,1)$
    are even;
  \item {\em nearly standard} if $\dd{c}(a,0) = \dd{c}(a,1) = (0,1)$ for
    almost all $a\in A$, except that there is at most one $a_1\in A$
    such that $\dd{c}(a_1,0)=(0,0)$ and $\dd{c}(a_1,1)=(1,1)$, and at most
    one $a_2\in A$ such that $\dd{c}(a_2,0)=(0,1)$ and
    $\dd{c}(a_2,1)=(1,0)$;
  \item {\em nearly symmetric} if $\dd{c}(a,0) = \dd{c}(a,1)$ for almost all
    $a\in A$, except that there is at most one $a_1\in A$ such that
    $\dd{c}(a_1,0)=(0,0)$ and $\dd{c}(a_1,1)=(1,1)$, and at most one $a_2\in
    A$ such that $\dd{c}(a_2,0)=(0,1)$ and $\dd{c}(a_2,1)=(1,0)$.
  \end{itemize}
\end{definition}

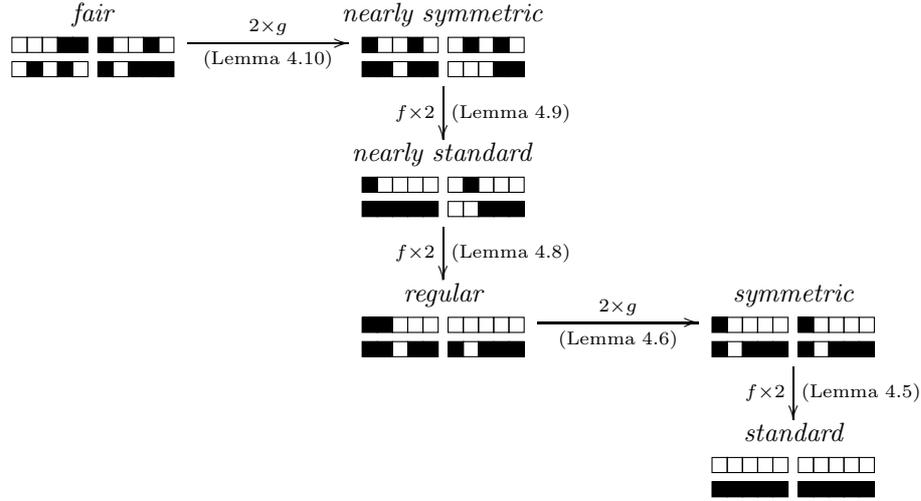
\begin{figure}
\[ 
\xymatrix@C+3.6em@R-1ex{
  *{
    \begin{array}{c}
      \textit{fair}\\
      \squares{11100}~\squares{01101}\\[-1mm]
      \squares{10101}~\squares{01000}
    \end{array}
  }
  \ar[r]^<>(.5){2\times g}_<>(.5){\text{(Lemma~\ref{lem-nearly-sym})}}
  & *{
    \begin{array}{c}
      \makebox[0in]{\textit{nearly symmetric}}\\
      \squares{01101}~\squares{10101}\\[-1mm]
      \squares{00100}~\squares{11100}
    \end{array}
  }
  \ar[d]_<>(.5){f\times 2}^<>(.5){\text{(Lemma~\ref{lem-nearly-standard})}}\\
  & *{
    \begin{array}{c}
      \makebox[0in]{\textit{nearly standard}}\\
      \squares{01111}~\squares{10111}\\[-1mm]
      \squares{00000}~\squares{11000}
    \end{array}
  }
  \ar[d]_<>(.5){f\times 2}^<>(.5){\text{(Lemma~\ref{lem-regular})}}\\
  & *{
    \begin{array}{c}
      \textit{regular}\\
      \squares{00111}~\squares{11111}\\[-1mm]
      \squares{00100}~\squares{01000}
    \end{array}
  }
  \ar[r]^<>(.5){2\times g}_<>(.5){\text{(Lemma~\ref{lem-sym})}}
  & *{
    \begin{array}{c}
      \textit{symmetric}\\
      \squares{01111}~\squares{01111}\\[-1mm]
      \squares{01000}~\squares{01000}
    \end{array}
  }
  \ar[d]_<>(.5){f\times 2}^<>(.5){\text{(Lemma~\ref{lem-standard})}}\\
  && *{
    \begin{array}{c}
      \textit{standard}\\
      \squares{11111}~\squares{11111}\\[-1mm]
      \squares{00000}~\squares{00000}
    \end{array}
  }
}
\]
\caption{Standardizing a fair permutation}\label{fig-standardizing}
\end{figure}

An example of each of these properties is shown in
Figure~\ref{fig-standardizing}.  Our strategy for proving
Proposition~\ref{pro-construction-2} is to use the action of
permutations of the forms $2\times g$ and $f\times 2$ to successively
improve the properties of a coloring until it is standard.  This
procedure is also outlined in Figure~\ref{fig-standardizing}, along
with the number of the lemma that will be used in each step. The
remainder of this section is devoted to the statements and proofs of
these lemmas, culminating in the proof of
Proposition~\ref{pro-construction-2} in Section~\ref{ssec-proof}.

\begin{lemma}\label{lem-standard}\label{lem-4.5}
  Let $c$ be a symmetric fair coloring of $2\times A\times 2$.  Then
  there exists $f\in\S(2\times A)$ such that $(f\times 2)\bu c$ is
  standard.
\end{lemma}

\begin{proof}
  Since $c$ is symmetric, we have $c(i,a,0) = c(i,a,1)$ for all
  $(i,a)\in 2\times A$; write $p(i,a) = c(i,a,0)$. Since $c$ is fair,
  $p : 2\times A \to 2$ is also fair. By Remark~\ref{rem-fair}, there
  exists a permutation $f \in \S(2\times A)$ such that $f\bu p$ is the
  standard coloring of $2\times A$. It follows that $(f\times 2)\bu c$
  is the standard coloring of $2\times A\times 2$.
\end{proof}

\begin{lemma}\label{lem-sym}
  Let $c$ be a regular coloring of $2\times A\times 2$. Then there
  exists $g\in\S(A\times 2)$ such that $(2\times g)\bu c$ is symmetric.
\end{lemma}

\begin{proof}
  Since $c$ is regular, we can find integers $p,q,r,s$ such that
  $n_c(0,0)=2p$, $n_c(1,1)=2q$, $n_c(0,1)=2r$, and $n_c(1,0)=2s$.
  Note that $n_c(0,0)+n_c(0,1)+n_c(1,0)+n_c(1,1)=2\card{A}$, and
  therefore $p+q+r+s = \card{A}$. Write $A$ as a disjoint union of
  sets $P\cup Q\cup R\cup S$, where $\card{P}=p$, $\card{Q}=q$,
  $\card{R}=r$, and $\card{S}=s$. Define a coloring $c'$ by
  \[
  \begin{array}{cl}
    \dd{c'}(a,0) = \dd{c'}(a,1) = (0,0), & \mbox{if $a\in P$,} \\
    \dd{c'}(a,0) = \dd{c'}(a,1) = (1,1), & \mbox{if $a\in Q$,} \\
    \dd{c'}(a,0) = \dd{c'}(a,1) = (0,1), & \mbox{if $a\in R$,} \\
    \dd{c'}(a,0) = \dd{c'}(a,1) = (1,0), & \mbox{if $a\in S$.} \\
  \end{array}
  \]
  Then by construction, $c'$ is symmetric and has the same color pair
  distribution as $c$. Hence by Lemma~\ref{lem-color-pair-dist},
  there exists $g\in\S(A\times 2)$ such that $c'=(2\times g)\bu c$,
  which was to be shown.
\end{proof}

\begin{lemma}\label{lem-regular-3}
  Suppose $\card{A}=3$ and $c$ is a nearly standard coloring of
  $2\times A\times 2$. Then there exists $f\in\S(2\times A)$ such that
  $(f\times 2)\bu c$ is regular.
\end{lemma}

\begin{proof}
  Write $A$ as the disjoint union $A_1\cup A_2\cup A_3$, where
  $\dd{c}(a_1,0)=(0,0)$ and $\dd{c}(a_1,1)=(1,1)$ for all $a_1\in A_1$,
  $\dd{c}(a_2,0)=(0,1)$ and $\dd{c}(a_2,1)=(1,0)$ for all $a_2\in A_2$, and
  $\dd{c}(a,0) = \dd{c}(a,1) = (0,1)$ for all $a\in A_3$. By the definition
  of nearly standard, we know that $A_1$ and $A_2$ have at most one
  element each. Since we assumed $\card{A}=3$, this leaves us with
  four cases. 
  \begin{itemize}
  \item Case 1. Assume $\card{A_1}=\card{A_2}=0$. Say
    $A_3=\s{p,q,r}$. Since $\dd{c}(a,0) = \dd{c}(a,1) = (0,1)$ for all
    $a\in A$, $c$ is the following coloring (using the notation
    of Section~\ref{ssec-vis-col}):
    \[
    \begin{array}{ccc|ccl}
      p & q & r & p & q & r \\
      \hline
      1 & 1 & 1 & 1 & 1 & 1 \\
      0 & 0 & 0 & 0 & 0 & 0. \\
    \end{array}
    \]
    Since $c$ is already regular (in fact, standard), we can take $f$
    to be the identity permutation.
  \item Case 2. Assume $\card{A_1}=1$ and $\card{A_2}=0$. Say
    $A_1=\s{a_1}$ and $A_3=\s{p,q}$. Then $c$ is the
    coloring
    \[
    \begin{array}{ccc|ccl}
      a_1 & p & q & a_1 & p & q \\
      \hline
      0 & 1 & 1 & 1 & 1 & 1 \\
      0 & 0 & 0 & 1 & 0 & 0. \\
    \end{array}
    \]
    Define $f : 2\times A \to 2\times A$ by $f(0,a_1)=(1,p)$,
    $f(1,p)=(0,q)$, $f(0,q)=(0,a_1)$, and the identity elsewhere.
    Then $(f\times 2)\bu c$ is the coloring
    \[
    \begin{array}{ccc|ccl}
      a_1 & p & q & a_1 & p & q \\
      \hline
      0 & 0 & 1 & 1 & 1 & 1 \\
      0 & 0 & 1 & 0 & 0 & 1. \\
    \end{array}
    \]
  \item Case 3. Assume $\card{A_1}=0$ and $\card{A_2}=1$. Say
    $A_2=\s{a_2}$ and $A_3=\s{p,q}$. Then $c$ is the coloring
    \[
    \begin{array}{ccc|ccl}
      a_2 & p & q & a_2 & p & q \\
      \hline
      1 & 1 & 1 & 0 & 1 & 1 \\
      0 & 0 & 0 & 1 & 0 & 0. \\
    \end{array}
    \]
    Define $f : 2\times A \to 2\times A$ by $f(0,a_2)=(1,p)$,
    $f(1,p)=(0,q)$, $f(0,q)=(0,a_2)$, and the identity elsewhere.
    Then $(f\times 2)\bu c$ is the coloring 
    \[
    \begin{array}{ccc|ccl}
      a_2 & p & q & a_2 & p & q \\
      \hline
      1 & 0 & 1 & 0 & 1 & 1 \\
      0 & 0 & 1 & 0 & 0 & 1. \\
    \end{array}
    \]
  \item Case 4. Assume $\card{A_1}=\card{A_2}=1$. Say $A_1=\s{a_1}$,
    $A_2=\s{a_2}$, and $A_3=\s{p}$. Then $c$ is the coloring 
    \[
    \begin{array}{ccc|ccl}
      a_1 & a_2 & p & a_1 & a_2 & p \\
      \hline
      0 & 1 & 1 & 1 & 0 & 1 \\
      0 & 0 & 0 & 1 & 1 & 0. \\
    \end{array}
    \]
    Define $f : 2\times A \to 2\times A$ by $f(0,a_1)=(1,a_2)$,
    $f(1,a_2)=(0,p)$, $f(0,p)=(0,a_1)$, and the identity elsewhere.
    Then $(f\times 2)\bu c$ is the coloring
    \[
    \begin{array}{ccc|ccl}
      a_1 & a_2 & p & a_1 & a_2 & p \\
      \hline
      0 & 0 & 1 & 1 & 1 & 1 \\
      0 & 0 & 1 & 0 & 1 & 0. \\
    \end{array}
    \]
  \end{itemize}
  In all four cases, $(f\times 2)\bu c$ is regular, as desired.
\end{proof}

\begin{lemma}\label{lem-regular}
  Suppose $\card{A}\geq 3$ and $c$ is a nearly standard coloring of
  $2\times A\times 2$. Then there exists $f\in\S(2\times A)$ such that
  $(f\times 2)\bu c$ is regular.
\end{lemma}

\begin{proof}
  The only difference to Lemma~\ref{lem-regular-3} is that $A$ may
  have more than 3 elements. However, by the definition of nearly
  standard, $c$ is already standard (hence regular) on the excess
  elements. Therefore, we can ignore all but 3 elements of $A$ and
  proceed as in Lemma~\ref{lem-regular-3}.
\end{proof}

\begin{lemma}\label{lem-nearly-standard}
  Let $c$ be a nearly symmetric fair coloring of $2\times A\times 2$.
  Then there exists $f\in\S(2\times A)$ such that $(f\times 2)\bu c$
  is nearly standard.
\end{lemma}

\begin{proof}
  Let $A'=\s{a\in A\mid \dd{c}(a,0)=\dd{c}(a,1)}$, and let $c'$ be the
  coloring of $2\times A'\times 2$ obtained by restricting $c$ to the
  domain $2\times A'\times 2$. Then $c'$ is symmetric. By definition
  of ``nearly symmetric'', there exists at most two elements
  $a_1,a_2\in A\setminus A'$; moreover, the element $a_1$, if any,
  satisfies $\dd{c}(a_1,0)=(0,0)$ and $\dd{c}(a_1,1)=(1,1)$ and the
  element $a_2$, if any, satisfies $\dd{c}(a_2,0)=(0,1)$ and
  $\dd{c}(a_2,1)=(1,0)$.  By assumption, $c$ is fair. Since $c$
  restricted to $2\times (A\setminus A')\times 2$ is evidently fair as
  well, it follows that $c'$ is also fair.  We will choose the
  permutation $f$ so that its support does not touch the elements
  $a_1$ and $a_2$; it therefore suffices to find some permutation
  $f'\in\S(2\times A')$ such that $(f'\times 2)\bu c'$ is standard.
  But such an $f'$ exists by Lemma~\ref{lem-standard}.
\end{proof}

\begin{lemma}\label{lem-nearly-sym}\label{lem-4.10}
  Let $c$ be a fair coloring of $2\times A\times 2$. Then there
  exists $g\in\S(A\times 2)$ such that $(2\times g)\bu c$ is nearly
  symmetric.
\end{lemma}

\begin{proof}
  The proof is very similar to that of Lemma~\ref{lem-sym}.  Start by
  considering the color pair distribution
  $(n_c(0,0), n_c(0,1), n_c(1,0), n_c(1,1))$ of the given coloring
  $c$, and note that
  \begin{equation}\label{eqn-sum}
    n_c(0,0)+n_c(0,1)+n_c(1,0)+n_c(1,1)=2\card{A}.
  \end{equation}
  Because $c$ is fair, we must have $n_c(0,0)=n_c(1,1)$, and in
  particular, $n_c(0,0)$ and $n_c(1,1)$ have the same parity (even or
  odd). From {\eqref{eqn-sum}}, it follows that $n_c(0,1)$ and
  $n_c(1,0)$ have the same parity. Therefore, there exist natural
  numbers $p,q,r,s,t,u$, with $t,u\in\s{0,1}$, such that
  \[ n_c(0,0) = 2p+t,\quad
  n_c(1,1) = 2q+t,\quad
  n_c(0,1) = 2r+u,\quad
  n_c(1,0) = 2s+u.
  \]
  (As a matter of fact, $p=q$, but we will not make further use of
  this fact).  From {\eqref{eqn-sum}}, we have that
  $p+q+r+s+t+u = \card{A}$.  Write $A$ as a disjoint union
  $P\cup Q\cup R\cup S\cup T\cup U$, where $\card{P}=p$, $\card{Q}=q$,
  $\card{R}=r$, $\card{S}=s$, $\card{T}=t$, and $\card{U}=u$.  Define
  a coloring $c'$ by
  \[ 
  \begin{array}{cl}
    \dd{c'}(a,0) = \dd{c'}(a,1) = (0,0), & \mbox{if $a\in P$,} \\
    \dd{c'}(a,0) = \dd{c'}(a,1) = (1,1), & \mbox{if $a\in Q$,} \\
    \dd{c'}(a,0) = \dd{c'}(a,1) = (0,1), & \mbox{if $a\in R$,} \\
    \dd{c'}(a,0) = \dd{c'}(a,1) = (1,0), & \mbox{if $a\in S$,} \\
    \dd{c'}(a,0) = (0,0) \quad\mbox{and}\quad \dd{c'}(a,1) = (1,1), & \mbox{if $a\in T$,} \\
    \dd{c'}(a,0) = (0,1) \quad\mbox{and}\quad \dd{c'}(a,1) = (1,0), & \mbox{if $a\in U$.} \\
  \end{array}                                            
  \]
  By construction, $c'$ has the same color pair distribution as $c$.
  Hence by Lemma~\ref{lem-color-pair-dist}, there exists
  $g\in\S(A\times 2)$ such that $c'=(2\times g)\bu c$.
  On the other hand, by construction, $c'$ is nearly symmetric (with
  $a_1$ being the unique element of $T$, if any, and $a_2$ being the
  unique element of $U$, if any). 
\end{proof}

\subsection{Proof of Proposition~\ref{pro-construction-2}}\label{ssec-proof}

Proposition~\ref{pro-construction-2} is now an easy consequence of
Lemmas~\ref{lem-standard}--\ref{lem-nearly-sym}. Figure~\ref{fig-standardizing}
contains a proof ``without words''. For readers who prefer a proof
``with words'', we give it here:

Assume $\card{A}\geq 3$ and let $c$ be a fair coloring of
$2\times A\times 2$. By Lemma~\ref{lem-nearly-sym}, there exists
$g_1\in\S(A\times 2)$ such that $c_1 = (2\times g_1)\bu c$ is nearly
symmetric. By Lemma~\ref{lem-nearly-standard}, there exists
$f_2\in\S(2\times A)$ such that $c_2=(f_2\times 2)\bu c_1$ is nearly
standard. By Lemma~\ref{lem-regular}, there exists
$g_3\in\S(A\times 2)$ such that $c_3=(2\times g_3)\bu c_2$ is regular.
By Lemma~\ref{lem-sym}, there exists $g_4\in\S(A\times 2)$ such that
$c_4=(2\times g_4)\bu c_3$ is symmetric. By Lemma~\ref{lem-standard},
there exists $f_5\in\S(2\times A)$ such that
$c_5=(f_5\times 2)\bu c_4$ is standard.  Let
\[ 
\begin{array}{rcl}
  \sigma &=& (f_5\times 2)(2\times g_4)(2\times g_3)(f_2\times
           2)(2\times g_1) \\
       &=& (f_5\times 2)(2\times g_4g_3)(f_2\times 2)(2\times g_1).
\end{array}
\]
Then $\sigma\bu c = \cstan$, and $\sigma$ has alternation depth 4, as
claimed. \qed

\section{Proof of the main theorem}\label{sec-proof-main}

Our main result, Theorem~\ref{thm-main}, follows from
Propositions~\ref{pro-construction-1} and {\ref{pro-construction-2}}.
Specifically, let $\sigma\in \S(2\times A\times 2)$ be an even
permutation, and let $c = \sigma\inv \bu \cstan$. By
Proposition~\ref{pro-construction-2}, we can find
$\tau\in\S(2\times A\times 2)$ of alternation depth 4 such that
$\tau\bu c = \cstan$. Note that $\tau$ is even by
Remark~\ref{rem-even}. Let $\rho = \sigma\tau\inv$. Then $\rho$ is
also even, and
$\rho\bu\cstan = \sigma\bu(\tau\inv \bu \cstan) = \sigma\bu c =
\cstan$.
Therefore, by Lemma~\ref{lem-glue}, $\rho$ is of the form $g+h$, for
$g,h\in\S(A\times 2)$. By Proposition~\ref{pro-construction-1}, $\rho$
has alternation depth 5, and it follows that $\sigma=\rho\tau$ has
alternation depth 9, as claimed. \qed

\section{Even alternation depth}\label{sec-even}

As promised in Remark~\ref{rem-even-main}, we now refine
Theorem~\ref{thm-main} to the case where in each factor of the form
$f\times 2$ and $2\times g$, the permutations $f$ and $g$ are required
to be even.

\begin{definition}
  We say that $\sigma\in\S(2\times A\times 2)$ has {\em even
    alternation depth} $\dee$ if it can be written as a product of
  $\dee$ factors of the forms $f\times 2$ or $2\times g$, where each
  such $f\in\S(2\times A)$ and $g\in\S(A\times 2)$ is an even
  permutation.
\end{definition}

We have the following analogue of Theorem~\ref{thm-main} for even
alternation depth:

\begin{theorem}\label{thm-even-main}
  Let $A$ be a finite set of 3 or more elements. Then every even
  permutation $\sigma \in \S(2\times A\times 2)$ has even alternation
  depth 13.
\end{theorem}

We prove Theorem~\ref{thm-even-main} by appropriately modifying the
proof of Theorem~\ref{thm-main}. We first prove that the permutation
$\sigma$ in Proposition~\ref{pro-construction-2} can be chosen to be
of even alternation depth 4, and then that the permutation $g+h$ in
Proposition~\ref{pro-construction-1} has even alternation depth 9.

\begin{lemma}\label{lem-odd-colorpairs}
  Let $\card{A}\geq 3$, and let $c$ be any coloring of $2\times A\times 2$.
  Then there exists an odd permutation $g\in\S(A\times 2)$ such that
  $(2\times g)\bu c=c$. Similarly, there exists an odd permutation
  $f\in\S(2\times A)$ such that $(f\times 2)\bu c = c$.
\end{lemma}

\begin{proof}
  Recall from Section~\ref{ssec-colorpair} that we can associate to
  each element $x\in A\times 2$ a color pair $\dd{c}(x)\in 2\times 2$.
  Since there are more elements in $A\times 2$ than in $2\times 2$,
  there must exist two elements $x,y\in A\times 2$ with the same color
  pair.  Let $g$ be the permutation that exchanges $x$ and $y$. Then
  clearly $g$ is odd and $g\bu c=c$. The second claim is proved
  symmetrically.
\end{proof}

\begin{proposition}\label{pro-even-construction-2}
  The permutation $\sigma$ in Proposition~\ref{pro-construction-2} can
  be chosen to be of even alternation depth 4.
\end{proposition}

\begin{proof}
  Using Lemma~\ref{lem-odd-colorpairs}, each of the permutations $f$ or
  $g$ constructed in Lemmas~\ref{lem-4.5}--\ref{lem-4.10} can be
  chosen to be even without affecting the coloring. Hence, all
  relevant permutations in the proof of
  Proposition~\ref{pro-construction-2} can be chosen to be even.
\end{proof}

We now turn to generalizing the results of Section~\ref{sec-first} to
the evenly alternating case. 

\begin{definition}
  A permutation $\sigma$ is {\em evenly balanced} if it can be written
  as the disjoint product of two even permutations $\tau$ and $\tau'$,
  such that $\tau$ and $\tau'$ are similar.
\end{definition}

We note that every evenly balanced permutation is balanced, but not
vice versa. For example, the permutation $(1\;2)(3\;4)$ is balanced,
but not evenly balanced.  The permutation $(1\;2\;3)(4\;5\;6)$ is
evenly balanced. A permutation is evenly balanced if and only if it is
balanced and the total number of odd cycles in its cycle decomposition
is divisible by $4$.

We have the following generalization of
Proposition~\ref{pro-balanced}:

\begin{lemma}\label{lem-supplement}
  Let $\sigma$ be an even permutation in $\S(X)$, where $\card{X}\geq 6$. 
  Then one of the following two conditions holds:
  \begin{enumerate}\alphalabels
  \item $\card{X}=6$ and $\sigma$ is the disjoint product of a
    $2$-cycle and a $4$-cycle; or
  \item there exist evenly balanced permutations $\rho,\tau$ such
    that $\sigma=\tau\rho$.
  \end{enumerate}
\end{lemma}

The proof of Lemma~\ref{lem-supplement} is a long and tedious case
distinction, and in the interest of readability, we give it separately
in Section~\ref{sec-supplement} below.

\begin{lemma}\label{lem-balanced-commutes-odd}
  Let $\card{X}\geq 2$. Every balanced permutation on $X$ commutes
  with some odd permutation.
\end{lemma}

\begin{proof}
  Let $\sigma\in\S(X)$ be a balanced permutation. If $\sigma$ is the
  identity, it commutes with everything, and therefore, since
  $\card{X}\geq 2$, with some odd permutation. Otherwise, by
  definition of being balanced, the cycle decomposition of $\sigma$
  contains two $\kay$-cycles $\gamma_1=(a_1\;a_2\;\ldots\;a_\kay)$ and
  $\gamma_2=(b_1\;b_2\;\ldots\;b_\kay)$, for some $\kay\geq 2$.  If
  $\kay$ is even, let $\tau=\gamma_1$. If $\kay$ is odd, let
  $\tau=(a_1\;b_1)(a_2\;b_2)\cdots(a_\kay\;b_\kay)$. In either case,
  $\tau$ is an odd permutation and commutes with $\sigma$.
\end{proof}

The next two lemmas generalize Lemmas~\ref{lem-balanced} and
{\ref{lem-balanced2}}, respectively.

\begin{lemma}\label{lem-even-balanced}
  Let $\tau\in\S(A\times 2)$ be an evenly balanced permutation. Then
  there exist even permutations $g\in\S(A\times 2)$ and $h\in\S(A)$
  such that $\tau = g\inv(h\times 2)g$.
\end{lemma}

\begin{proof}
  By Lemma~\ref{lem-balanced}, we can find $g'\in\S(A\times 2)$ and
  $h\in\S(A)$ such that $\tau = g'^{-1}(h\times 2)g'$. Since $\tau$ is
  evenly balanced, so is $h\times 2$, and it follows that $h$ is even.
  If $g'$ is even, let $g=g'$, and we are done. If $g'$ is odd, let
  $g''$ be some odd permutation that commutes with $h\times 2$; such a
  $g''$ exists by Lemma~\ref{lem-balanced-commutes-odd}. Letting
  $g=g''g'$, we have
  $g\inv(h\times 2)g = g'^{-1}g''^{-1}(h\times 2)g''g' =
  g'^{-1}(h\times 2)g' = \tau$.
\end{proof}

\begin{lemma}\label{lem-even-balanced2}
  Let $\tau\in\S(A\times 2)$ be an evenly balanced permutation, and
  let $\sigma=\id_{A\times 2}+\tau\in\S(2\times A\times 2)$. Then
  there exist even permutations $g\in\S(A\times 2)$ and
  $f\in\S(2\times A)$ such that
  \[ \sigma = (2\times g\inv)(f\times 2)(2\times g).
  \]
\end{lemma}

\begin{proof}
  Exactly like the proof of Lemma~\ref{lem-balanced2}, except using
  Lemma~\ref{lem-even-balanced} instead of Lemma~\ref{lem-balanced}.
\end{proof}

Next, we deal with two special cases.

\begin{lemma}\label{lem-2+4}
  Let $\card{A}\geq3$ and let $\tau\in\S(A\times 2)$ be the disjoint
  product of a 2-cycle and a 4-cycle. Then
  $\id+\tau \in \S(2\times A\times 2)$ has even alternation depth 5,
  starting and ending with a factor of the form $2\times g$.
\end{lemma}

\begin{proof}
  Let $a,b,c$ be three distinct elements of $A$. First consider the
  permutation $\tau'\in\S(A\times 2)$ that is given, in cycle
  notation, by
  \[ \tau' = \cyc{(a,0),(a,1)}\;\cyc{(b,0),(c,0),(c,1),(b,1)}.
  \]
  Let 
  \[ 
  \begin{array}{rclll}
    g_1 &=& \cyc{(b,1), (c,0), (c,1)} &\in& \S(A\times 2),\\
    f_2 &=& \cyc{(0,c), (1,c), (1,b)} &\in& \S(2\times A),\\
    g_3 &=& \cyc{(a,0), (a,1), (b,1)} &\in& \S(A\times 2),\\
    f_4 &=& \cyc{(0,b), (0,c), (1,b), (1,a), (1,c)} &\in& \S(2\times A),\\
    g_5 &=& \cyc{(a,0), (b,1), (c,1)} &\in& \S(A\times 2),
  \end{array}
  \]
  all of which are even. By a direct calculation, one can verify that
  \begin{equation}\label{eqn-id+rho}
    \id+\tau' = (2\times g_5)(f_4\times 2)                 
    (2\times g_3)(f_2\times 2)(2\times g_1),
  \end{equation}
  and therefore $\id+\tau'$ has been written in the desired form. Now
  consider some arbitrary $\tau\in\S(A\times 2)$ that is the disjoint
  product of a 2-cycle and a 4-cycle. Since $\tau$ is similar to $\tau'$,
  there is some $g\in\S(A\times 2)$ such that $\tau=g\inv\tau' g$.
  Without loss of generality, we can assume that $g$ is even, because
  if it is not, we can replace $g$ by $\cyc{(a,0),(a,1)}g$, noting
  that $\cyc{(a,0),(a,1)}$ is odd and commutes with $\tau'$. Then
  $\id+\tau = (2\times g\inv)(\id+\tau')(2\times g)$. When combined
  with {\eqref{eqn-id+rho}}, this shows that $\id+\tau$ is of the
  required form.
\end{proof}

\begin{lemma}\label{lem-odd-even}
  Let $\card{A}\geq 2$. There exist permutations
  $g_1,g_3,g_5\in\S(A\times 2)$ and $f_2,f_4,f_6\in\S(2\times A)$ such
  that $g_1$ is odd, $f_2,g_3,f_4,g_5,f_6$ are even, and such that 
  \[ \id_{2\times A\times 2} = (f_6\times 2)(2\times g_5)(f_4\times
  2)(2\times g_3)(f_2\times 2)(2\times g_1).
  \]
\end{lemma}

\begin{proof}
  Let $a,b$ be distinct elements of $A$. The following permutations,
  given in cycle notation, satisfy the conclusion of the lemma:
  \[
  \begin{array}{rclll}
    g_1 &=& \cyc{(a,0), (b,1)} &\in& \S(A\times 2),\\
    f_2 &=& \cyc{(1,b),(0,a),(1,a)} &\in& \S(2\times A),\\
    g_3 &=& \cyc{(b,0), (a,0), (b,1)} &\in& \S(A\times 2),\\
    f_4 &=& \cyc{(1,b),(1,a),(0,a)} &\in& \S(2\times A),\\
    g_5 &=& \cyc{(b,0), (a,0), (b,1)} &\in& \S(A\times 2),\\
    f_6 &=& \cyc{(0,b),(1,b),(1,a)} &\in& \S(2\times A).
  \end{array}
  \]
\end{proof}

The following two propositions generalize
Proposition~\ref{pro-construction-1+h} and {\ref{pro-construction-1}},
respectively, to the evenly alternating case.

\begin{proposition}\label{pro-even-construction-1+h}
  Let $A$ be a finite set of 3 or more elements, and let
  $\tau\in\S(A\times 2)$ be an even permutation. Then
  $\id_{A\times 2}+\tau \in \S(2\times A\times 2)$ has an evenly
  alternating decomposition of depth 5, starting and ending with a
  factor of the form $2\times g$.
\end{proposition}

\begin{proof}
  If $\tau$ is the disjoint product of a $2$-cycle and a $4$-cycle, the
  claim follows by Lemma~\ref{lem-2+4}. Otherwise, by
  Lemma~\ref{lem-supplement}, there exist evenly balanced permutations
  $\tau_1,\tau_2\in\S(A\times 2)$ such that $\tau=\tau_2\tau_1$. By
  Lemma~\ref{lem-even-balanced2}, there exist even
  $g_1,g_2\in\S(A\times 2)$ and $f_1,f_2\in\S(2\times A)$ such that
  $\id+\tau_i = (2\times g_i\inv)(f_i\times 2)(2\times g_i)$, for
  $i=1,2$. Therefore
  \[ \id+\tau = (2\times g_2\inv)(f_2\times 2)(2\times g_2g_1\inv)               
  (f_1\times 2)(2\times g_1),
  \]
  proving the claim.
\end{proof}

\begin{proposition}\label{pro-even-construction-1}
  Let $A$ be a finite set of 3 or more elements, and let
  $g,h\in\S(A\times 2)$ be permutations such that $\sigma=g+h$ is
  even. Then $\sigma$ has even alternation depth 9.
\end{proposition}

\begin{proof}
  Let $\tau=hg\inv\in\S(A\times 2)$, and note that $\tau$ is even. By
  Proposition~\ref{pro-even-construction-1+h}, $\id+\tau$ can be
  written in the form $(2\times g_5)(f_4\times 2)                 
  (2\times g_3)(f_2\times 2)(2\times g_1)$, where all the $f_i$ and
  $g_j$ are even. As in the proof of
  Proposition~\ref{pro-construction-1}, we have
  \begin{equation}\label{eqn-pro-even-construction-1a}
    \sigma = g+h = (2\times g_5)(f_4\times 2)(2\times g_3)(f_2\times
    2)(2\times g_1g).
  \end{equation}
  If $g$ is even, we are done; in fact, in this case, $\sigma$ has
  even alternation depth 5. So assume $g$ is odd. By
  Lemma~\ref{lem-odd-even}, there exists an odd permutation
  $g_6\in\S(A\times 2)$ and even permutations
  $f_7,f_9,f_{11}\in\S(2\times A)$ and $g_8,g_{10}\in\S(A\times 2)$ such
  that
  \begin{equation}\label{eqn-pro-even-construction-1b}
    \id = (f_{11}\times 2)(2\times g_{10})(f_9\times 2)(2\times g_8)(f_7\times
    2)(2\times g_6).
  \end{equation}
  Combining {\eqref{eqn-pro-even-construction-1a}} and
  {\eqref{eqn-pro-even-construction-1b}}, we have
  \[ \sigma = (2\times g_5)(f_4\times 2)(2\times g_3)(f_2f_{11}\times
  2)(2\times g_{10})(f_9\times 2)(2\times g_8)(f_7\times 2)(2\times
  g_6g_1g),
  \]
  and since $g_6g_1g$ is even, this demonstrates that $\sigma$ has
  even alternation depth 9.
\end{proof}

Finally, Theorem~\ref{thm-even-main} follows from
Propositions~\ref{pro-even-construction-2} and
{\ref{pro-even-construction-1}} in exactly the same way as
Theorem~\ref{thm-main} was shown to follow from
Propositions~\ref{pro-construction-2} and {\ref{pro-construction-1}}
in Section~\ref{sec-proof-main}.

\section{Proof of Lemma~\ref{lem-supplement}}\label{sec-supplement}

Recall that a cycle of length $\kay$ is an even permutation if and
only if $\kay$ is odd. We denote a conjugacy class of permutations by
$\kay_1\cup\ldots\cup \kay_n$, where $\kay_1,\ldots,\kay_n$ are the
lengths of the cycles in its cycle decomposition. We further
schematically write $\plus{\kay}$ for any cycle of length $\kay+4i$,
where $i\in\s{0,1,2,\ldots}$.

We say that an even non-identity permutation is {\em atomic} if it
cannot be factored into disjoint smaller even non-identity
permutations. Clearly every even permutation can be factored into
disjoint atoms. The atomic permutations are exactly: (a) odd-length
cycles, and (b) disjoint products of two even-length cycles. Using the
above notations, we can further divide the atomic permutations into 14
cases:
\begin{equation}\label{eqn-atoms}
\underline{3},\,
\underline{5},\,
\plus{7},\,
\plus{9},\,
\underline{2\cup 2},\,
\underline{2\cup 4},\,
2\cup\plus{6},\,
2\cup\plus{8},\,
4\cup 4,\,
4\cup\plus{6},\,
4\cup\plus{8},\,
\plus{6}\cup\plus{6},\,
\plus{6}\cup\plus{8},\,
\plus{8}\cup\plus{8}.
\end{equation}
An atom is called {\em special} if it has been underlined in
{\eqref{eqn-atoms}}. Therefore, the four special atoms are
$\underline{3}$, $\underline{5}$, $\underline{2\cup 2}$, and
$\underline{2\cup 4}$. An even permutation is {\em decomposable} if it
can be written as a product of two evenly balanced permutations.

\begin{proposition}\label{pro-decomposable}
  Every even permutation, other than a special atom, is decomposable.
  Moreover, if the size of the domain is $6$ or greater, the special
  atoms $\underline{3}$, $\underline{5}$, $\underline{2\cup 2}$ are
  decomposable. Moreover, if the size of the domain is $7$ or greater,
  the special atom $\underline{2\cup 4}$ is decomposable.
\end{proposition}

\begin{proof}
  Let us schematically write $A$ for a non-special atom and $\uA$ for
  a special atom. To prove the first claim of
  Proposition~\ref{pro-decomposable}, it is sufficient to show that
  the following are decomposable:
  \[
  (1)~ A,\quad
  (2)~ \uA\cup\uA,\quad
  (3)~ \uA\cup A,\quad
  (4)~ \uA\cup\uA\cup\uA.
  \]
  This is because every even permutation (except special atoms) can be
  decomposed into disjoint factors of the forms (1)--(4). (This
  follows by an easy induction on the number of special atoms). To
  prove the second and third claims of
  Proposition~\ref{pro-decomposable}, we must moreover show that the
  following is decomposable when the size of the domain is large
  enough:
  \[ (0)~ \uA.
  \]
  Since each atom is of one of the fourteen types shown in
  {\eqref{eqn-atoms}}, each claim (0)--(4) amounts to a finite case
  distinction. We exhibit a particular decomposition into evenly
  balanced permutations for each of the cases. 

  For greater brevity, we use the following convention in the cycle
  notation for permutations. Let $i,j\in\s{0,1,2,\ldots}$. We write
  $\bulA{n}$ as a shorthand for a list of $2i+1$ elements
  $n_1,\ldots,n_{2i+1}$, and we similarly write $\bulB{m}$ as a
  shorthand for a list of $2j+1$ elements $m_1,\ldots,m_{2j+1}$. Thus,
  $(\bulA{1},2,3)$ denotes the $(2i+3)$-cycle
  $(1_1,\ldots,1_{2i+1},2,3)$, and $(\bulA{1},\bulB{2},3)$ denotes the
  $(2i+2j+3)$-cycle
  $(1_1,\ldots,1_{2i+1},2_1,\ldots,2_{2j+1},3)$. Here, it is
  understood that multiple $\bulA{n}$ notations appearing in the same
  equation share the same value of $i$, and multiple $\bulB{m}$
  notations share the same value of $j$.

  \flushleft
  \begin{enumerate}
  \item[(0)] {\bf Decomposition of $\uA$ (4 cases)}

    \sethead{$\underline{2\cup2}$, where $\card{X}> 6$:}

    \begin{itemize}
    \item \head{$\underline{3}$, where $\card{X}\geq 6$:}
      \triple{$(1,2,3)$}{$(1,3,2)(4,5,6)$}{$(1,3,2)(4,6,5)$.}
    \item \head{$\underline{5}$, where $\card{X}\geq 6$:}
      \triple{$(1,2,3,4,5)$}{$(1,4,3)(2,6,5)$}{$(1,5,3)(2,4,6)$.}
    \item \head{$\underline{2\cup2}$, where $\card{X}\geq 6$:}
      \triple{$(1,2)(3,4)$}{$(1,3,5)(2,4,6)$}{$(1,6,4)(2,5,3)$.}
    \item \head{$\underline{2\cup4}$, where $\card{X}\geq 7$:}
      \triple{$(1,2)(3,4,5,6)$}{$(2,4,3)(5,6,7)$}{$(1,3,2)(4,7,6)$.}
    \end{itemize}

  \item[(1)] {\bf Decomposition of $A$ (10 cases)}

    \sethead{$\plus{8}\cup\plus{8}$:}

    \begin{itemize}
    \item \head{$\plus{7}$:}
      \triple{$(\bulA{1},2,3,\bulA{4},5,6,7)$}{$(1,6,5)(2,3,7)$}{$(\bulA{1},7,5)(\bulA{4},6,3)$.}
    \item \head{$\plus{9}$:}
      \triple{$(\bulA{1},2,3,\bulA{4},5,6,7,8,9)$}{$(2,3,8)(5,6,7)$}{$(\bulA{1},8,9)(\bulA{4},7,3)$.}
    \item \head{$2\cup\plus{6}$:}
      \triple{$(1,2)(\bulA{3},\bulA{4},5,6,7,8)$}{$(3,1,4)(5,6,7)$}{$(\bulA{3},1,2)(\bulA{4},7,8)$.}
    \item \head{$2\cup\plus{8}$:}
      \triple{$(1,2)(\bulA{3},4,5,\bulA{6},7,8,9,10)$}{$(3,2,6)(7,9,8)$}{$(\bulA{3},4,5,2,1)(\bulA{6},8,7,9,10)$.}
    \item \head{$4\cup4$:}
      \triple{$(1,2,3,4)(5,6,7,8)$}{$(1,2,4)(3,6,5)$}{$(2,5,3)(6,7,8)$.}
    \item \head{$4\cup\plus{6}$:}
      \triple{$(1,2,3,4)(\bulA{5},\bulA{6},7,8,9,10)$}{$(2,6,5)(7,9,8)$}{$(1,\bulA{5},2,3,4)(\bulA{6},8,7,9,10)$.}
    \item \head{$4\cup\plus{8}$:}
      \triple{$(1,2,3,4)(\bulA{5},\bulA{6},7,8,9,10,11,12)$}{$(2,6,5)(8,9,10)$}{$(1,\bulA{5},2,3,4)(\bulA{6},7,10,11,12)$.}
    \item \head{$\plus{6}\cup\plus{6}$:}
      \triple{$(\bulA{1},\bulA{2},3,4,5,6)(\bulB{7},8,9,\bulB{10},11,12)$}{$(2,5)(7,10)(8,11)(9,12)$}{$(\bulA{1},5,6)(\bulA{2},3,4)(\bulB{7},11,9)(\bulB{10},8,12)$.}
    \item \head{$\plus{6}\cup\plus{8}$:}
      \triple{$(\bulA{1},\bulA{2},3,4,5,6)(7,\bulB{8},\bulB{9},10,11,12,13,14)$}{$(2,5)(7,8)(9,13)(10,11)$}{$(\bulA{1},5,6)(\bulA{2},3,4)(\bulB{8},13,14)(\bulB{9},11,12)$.}
    \item \head{$\plus{8}\cup\plus{8}$:}
      \triple{$(1,2,3,4,5,6,\bulA{7},\bulA{8})(9,\bulB{10},\bulB{11},12,13,14,15,16)$}{$(1,2)(3,7,4)(5,6,8)(9,10)(11,15)(12,13)$}{$(2,4,\bulA{8})(3,\bulA{7},6)(\bulB{10},15,16)(\bulB{11},13,14)$.}
    \end{itemize}

  \item[(2)] {\bf Decomposition of $\uA\cup\uA$ (10 cases)}

    \sethead{$\underline{5}\cup(\underline{2\cup 2})$:}

    \begin{itemize}
    \item \head{$\underline{3}\cup\underline{3}$:}
      \triple{$(1,2,3)(4,5,6)$}{$\id$}{$(1,2,3)(4,5,6)$.}
    \item \head{$\underline{3}\cup\underline{5}$:}
      \triple{$(1,2,3)(4,5,6,7,8)$}{$(1,3,2)(6,7,8)$}{$(1,3,2)(4,5,8)$.}
    \item \head{$\underline{3}\cup(\underline{2\cup2})$:}
      \triple{$(1,2,3)(4,5)(6,7)$}{$(1,3,2)(4,7,6)$}{$(1,3,2)(4,5,6)$.}
    \item \head{$\underline{3}\cup(\underline{2\cup4})$:}
      \triple{$(1,2,3)(4,5)(6,7,8,9)$}{$(1,2,3)(4,7,6)$}{$(4,5,6)(7,8,9)$.}
    \item \head{$\underline{5}\cup\underline{5}$:}
      \triple{$(1,2,3,4,5)(6,7,8,9,10)$}{$\id$}{$(1,2,3,4,5)(6,7,8,9,10)$.}
    \item \head{$\underline{5}\cup(\underline{2\cup2})$:}
      \triple{$(1,2,3,4,5)(6,7)(8,9)$}{$(2,3,4)(6,8,9)$}{$(1,4,5)(6,7,9)$.}
    \item \head{$\underline{5}\cup(\underline{2\cup4})$:}
      \triple{$(1,2,3,4,5)(6,7)(8,9,10,11)$}{$(2,5)(3,4)(8,11)(9,10)$}{$(1,5)(2,4)(6,7)(8,10)$.}
    \item \head{$(\underline{2\cup2})\cup(\underline{2\cup2})$:}
      \triple{$(1,2)(3,4)(5,6)(7,8)$}{$\id$}{$(1,2)(3,4)(5,6)(7,8)$.}
    \item \head{$(\underline{2\cup2})\cup(\underline{2\cup4})$:}
      \triple{$(1,2)(3,4)(5,6)(7,8,9,10)$}{$(1,4)(2,3)(7,10)(8,9)$}{$(1,3)(2,4)(5,6)(7,9)$.}
    \item \head{$(\underline{2\cup4})\cup(\underline{2\cup4})$:}
      \triple{$(1,2)(3,4,5,6)(7,8)(9,10,11,12)$}{$\id$}{$(1,2)(3,4,5,6)(7,8)(9,10,11,12)$.}
    \end{itemize}

  \item[(3)] {\bf Decomposition of $\uA\cup A$ (40 cases)}

    \sethead{$\underline{3}\cup(\underline{2\cup \plus{6}})$:}

    \begin{itemize}
    \item \head{$\underline{3}\cup\plus{7}$:}
      \triple{$(1,2,3)(\bulA{4},5,6,\bulA{7},8,9,10)$}{$(1,2,3)(5,6,9)$}{$(\bulA{4},9,10)(6,\bulA{7},8)$.}
    \item \head{$\underline{3}\cup\plus{9}$:}
      \triple{$(1,2,3)(4,\bulA{5},6,7,8,9,\bulA{10},11,12)$}{$(2,6,4)(8,9,10)$}{$(1,4,\bulA{5},2,3)(6,7,\bulA{10},11,12)$.}
    \item \head{$\underline{3}\cup(2\cup\plus{6})$:}
      \triple{$(1,2,3)(4,5)(\bulA{6},7,8,\bulA{9},10,11)$}{$(1,2)(4,5)(7,11)(8,10)$}{$(2,3)(\bulA{6},11)(7,10)(8,\bulA{9})$.}
    \item \head{$\underline{3}\cup(2\cup\plus{8})$:}
      \triple{$(1,2,3)(4,5)(\bulA{6},7,8,\bulA{9},10,11,12,13)$}{$(1,2,3)(4,5,6)(7,8,13)(10,11,12)$}{$(5,\bulA{6},13)(8,\bulA{9},12)$.}
    \item \head{$\underline{3}\cup(4\cup4)$:}
      \triple{$(1,2,3)(4,5,6,7)(8,9,10,11)$}{$(1,2)(4,5)(6,7)(9,11)$}{$(2,3)(5,7)(8,11)(9,10)$.}
    \item \head{$\underline{3}\cup(4\cup\plus{6})$:}
      \triple{$(1,2,3)(4,5,6,7)(\bulA{8},9,\bulA{10},11,12,13)$}{$(1,2,3)(4,9,8)$}{$(4,5,6,7,\bulA{8})(9,\bulA{10},11,12,13)$.}
    \item \head{$\underline{3}\cup(4\cup\plus{8})$:}
      \triple{$(1,2,3)(4,5,6,7)(\bulA{8},9,\bulA{10},11,12,13,14,15)$}{$(1,2)(4,5)(6,7)(10,14)$}{$(2,3)(5,7)(\bulA{8},9,14,15)(\bulA{10},11,12,13)$.}
    \item \head{$\underline{3}\cup(\plus{6}\cup\plus{6})$:}
      \triple{$(1,2,3)(4,5,\bulA{6},\bulA{7},8,9)(\bulB{10},\bulB{11},12,13,14,15)$}{$(1,2,3)(4,5,6)(7,9)(12,13)(11,14)(10,15)$}{$(\bulA{6},9)(\bulA{7},8)(\bulB{10},14)(\bulB{11},13)$.}
    \item \head{$\underline{3}\cup(\plus{6}\cup\plus{8})$:}
      \triple{$(1,2,3)(\bulA{4},\bulA{5},6,7,8,9)(10,\bulB{11},\bulB{12},13,14,15,16,17)$}{$(1,2,3)(6,9,8)(5,7)(10,11)(12,16)(13,14)$}{$(\bulA{4},7,9)(\bulA{5},8,6)(\bulB{11},16,17)(\bulB{12},14,15)$.}
    \item \head{$\underline{3}\cup(\plus{8}\cup\plus{8})$:}
      \triple{$(1,2,3)(4,\bulA{5},\bulA{6},7,8,9,10,11)(\bulB{12},\bulB{13},14,15,16,17,18,19)$}{$(1,2)(4,5)(6,11)(7,10)(8,9)(13,19)(14,18)(15,17)$}{$(2,3)(\bulA{5},11)(\bulA{6},10)(7,9)(\bulB{12},19)(\bulB{13},18)(14,17)(15,16)$.}
    \item \head{$\underline{5}\cup\plus{7}$:}
      \triple{$(1,2,3,4,5)(6,7,8,9,10,\bulA{11},\bulA{12})$}{$(3,5,4)(6,11,12)$}{$(1,2,4,3,5)(6,7,8,9,10)(\bulA{11})(\bulA{12})$.}
    \item \head{$\underline{5}\cup\plus{9}$:}
      \triple{$(1,2,3,4,5)(\bulA{6},7,8,\bulA{9},10,11,12,13,14)$}{$(1,2,3,4,5)(7,8,11,12,13)$}{$(\bulA{6},13,14)(8,\bulA{9},10)$.}
    \item \head{$\underline{5}\cup(2\cup\plus{6})$:}
      \triple{$(1,2,3,4,5)(6,7)(8,9,\bulA{10},\bulA{11},12,13)$}{$(2,6,7,3,5)(9,10,11,12,13)$}{$(1,5)(2,7)(3,4)(8,13)(\bulA{10})(\bulA{11})$.}
    \item \head{$\underline{5}\cup(2\cup\plus{8})$:}
      \triple{$(1,2,3,4,5)(6,7)(8,9,\bulA{10},11,\bulA{12},13,14,15)$}{$(2,5)(3,4)(6,7)(8,12)$}{$(1,5)(2,4)(8,9,\bulA{10},11)(\bulA{12},13,14,15)$.}
    \item \head{$\underline{5}\cup(4\cup4)$:}
      \triple{$(1,2,3,4,5)(6,7,8,9)(10,11,12,13)$}{$(2,3,5)(11,13,12)$}{$(1,5)(3,4)(6,7,8,9)(10,12,11,13)$.}
    \item \head{$\underline{5}\cup(4\cup\plus{6})$:}
      \triple{$(1,2,3,4,5)(6,7,8,9)(10,\bulA{11},12,13,\bulA{14},15)$}{$(2,5)(3,4)(7,9)(10,13)$}{$(1,5)(2,4)(6,9)(7,8)(10,\bulA{11},12)(13,\bulA{14},15)$.}
    \item \head{$\underline{5}\cup(4\cup\plus{8})$:}
      \triple{$(1,2,3,4,5)(6,7,8,9)(10,\bulA{11},12,13,14,\bulA{15},16,17)$}{$(1,2)(4,5)(6,7,9,8)(10,13,14,17)$}{$(2,3,5)(7,9,8)(10,\bulA{11},12)(14,\bulA{15},16)$.}
    \item \head{$\underline{5}\cup(\plus{6}\cup\plus{6})$:}
      \triple{$(1,2,3,4,5)(\bulA{6},\bulA{7},8,9,10,11)(\bulB{12},\bulB{13},14,15,16,17)$}{$(2,3,5)(6,7,8)(15,17,16)(12,13,14)$}{$(1,5)(3,4)(8,9,10,11)(\bulA{6})(\bulA{7})(14,16,15,17)(\bulB{12})(\bulB{13})$.}
    \item \head{$\underline{5}\cup(\plus{6}\cup\plus{8})$:}
      \triple{$(1,2,3,4,5)(6,\bulA{7},8,9,\bulA{10},11)(\bulB{12},13,14,15,\bulB{16},17,18,19)$}{$(2,5)(3,4)(6,9)(12,16)$}{$(1,5)(2,4)(6,\bulA{7},8)(9,\bulA{10},11)(\bulB{12},13,14,15)(\bulB{16},17,18,19)$.}
    \item \head{$\underline{5}\cup(\plus{8}\cup\plus{8})$:}
      \triple{$(1,2,3,4,5)(\bulA{6},7,8,9,\bulA{10},11,12,13)(\bulB{14},15,16,\bulB{17},18,19,20,21)$}{$(1,2)(4,5)(7,9,10,12,8,13)(15,16,17,19,20,21)$}{$(2,3,5)(\bulA{6},13)(7,12,8)(\bulA{10},11)(\bulB{14},21)(\bulB{17},18)$.}
    \item \head{$(\underline{2\cup2})\cup\plus{7}$:}
      \triple{$(1,2)(3,4)(5,\bulA{6},\bulA{7},8,9,10,11)$}{$(1,2)(3,4)(5,6)(7,10)$}{$(\bulA{6},10,11)(\bulA{7},8,9)$.}
    \item \head{$(\underline{2\cup2})\cup\plus{9}$:}
      \triple{$(1,2)(3,4)(\bulA{5},6,7,\bulA{8},9,10,11,12,13)$}{$(1,4,3)(5,8,11)$}{$(1,2,3)(\bulA{5},6,7)(\bulA{8},9,10)(11,12,13)$.}
    \item \head{$(\underline{2\cup2})\cup(2\cup\plus{6})$:}
      \triple{$(1,2)(3,4)(5,6)(\bulA{7},8,9,\bulA{10},11,12)$}{$(1,2)(3,4)(5,6)(7,10)$}{$(\bulA{7},8,9)(\bulA{10},11,12)$.}
    \item \head{$(\underline{2\cup2})\cup(2\cup\plus{8})$:}
      \triple{$(1,2)(3,4)(5,6)(\bulA{7},8,9,10,\bulA{11},12,13,14)$}{$(1,4)(2,3)(7,10)(11,14)$}{$(1,3)(2,4)(5,6)(\bulA{7},8,9)(10,14)(\bulA{11},12,13)$.}
    \item \head{$(\underline{2\cup2})\cup(4\cup4)$:}
      same as $(\underline{2\cup4})\cup(\underline{2\cup4})$.
    \item \head{$(\underline{2\cup2})\cup(4\cup\plus{6})$:}
      \triple{$(1,2)(3,4)(5,6,7,8)(\bulA{9},10,11,\bulA{12},13,14)$}{$(1,3)(2,4)(6,8)(9,12)$}{$(1,4)(2,3)(5,8)(6,7)(\bulA{9},10,11)(\bulA{12},13,14)$.}
    \item \head{$(\underline{2\cup2})\cup(4\cup\plus{8})$:}
      \triple{$(1,2)(3,4)(5,6,7,8)(\bulA{9},10,11,12,\bulA{13},14,15,16)$}{$(1,2)(3,4)(5,6,7,8)(9,12,13,16)$}{$(\bulA{9},10,11)(\bulA{13},14,15)$.}
    \item \head{$(\underline{2\cup2})\cup(\plus{6}\cup\plus{6})$:}
      \triple{$(1,2)(3,4)(\bulA{5},6,7,\bulA{8},9,10)(\bulB{11},12,13,\bulB{14},15,16)$}{$(1,2)(3,4)(5,8)(11,14)$}{$(\bulA{5},6,7)(\bulA{8},9,10)(\bulB{11},12,13)(\bulB{14},15,16)$.}
    \item \head{$(\underline{2\cup2})\cup(\plus{6}\cup\plus{8})$:}
      \triple{$(1,2)(3,4)(\bulA{5},6,7,\bulA{8},9,10)(11,12,\bulB{13},\bulB{14},15,16,17,18)$}{$(1,4)(2,3)(5,8)(11,12,13)(14,18)(15,16,17)$}{$(1,3)(2,4)(\bulA{5},6,7)(\bulA{8},9,10)(\bulB{13},18)(\bulB{14},17)$.}
    \item \head{$(\underline{2\cup2})\cup(\plus{8}\cup\plus{8})$:}
      \triple{$(1,2)(3,4)(\bulA{5},6,7,8,\bulA{9},10,11,12)(\bulB{13},14,15,16,\bulB{17},18,19,20)$}{$(1,2)(3,4)(5,8,9,12)(13,16,17,20)$}{$(\bulA{5},6,7)(\bulA{9},10,11)(\bulB{13},14,15)(\bulB{17},18,19)$.}
    \item \head{$(\underline{2\cup4})\cup\plus{7}$:}
      \triple{$(1,2)(3,4,5,6)(\bulA{7},8,9,10,\bulA{11},12,13)$}{$(2,3,6)(7,10,11)$}{$(1,6,2)(3,4,5)(\bulA{7},8,9)(\bulA{11},12,13)$.}
    \item \head{$(\underline{2\cup4})\cup\plus{9}$:}
      \triple{$(1,2)(3,4,5,6)(7,\bulA{8},9,10,\bulA{11},12,13,14,15)$}{$(3,4)(5,6)(7,8)(11,15)$}{$(1,2)(4,6)(\bulA{8},9,10,15)(\bulA{11},12,13,14)$.}
    \item \head{$(\underline{2\cup4})\cup(2\cup\plus{6})$:}
      same as $(\underline{2\cup2})\cup(4\cup\plus{6})$.
    \item \head{$(\underline{2\cup4})\cup(2\cup\plus{8})$:}
      same as $(\underline{2\cup2})\cup(4\cup\plus{8})$.
    \item \head{$(\underline{2\cup4})\cup(4\cup4)$:}
      \triple{$(1,2)(3,4,5,6)(7,8,9,10)(11,12,13,14)$}{$(3,6,5)(7,8,9)$}{$(1,2)(3,4,6,5)(9,10)(11,12,13,14)$.}
    \item \head{$(\underline{2\cup4})\cup(4\cup\plus{6})$:}
      \triple{$(1,2)(3,4,5,6)(7,8,9,10)(\bulA{11},12,13,\bulA{14},15,16)$}{$(1,2)(3,4,5,6)(7,8,9,10)(11,14)$}{$(\bulA{11},12,13)(\bulA{14},15,16)$.}
    \item \head{$(\underline{2\cup4})\cup(4\cup\plus{8})$:}
      \triple{$(1,2)(3,4,5,6)(7,8,9,10)(\bulA{11},12,13,14,\bulA{15},16,17,18)$}{$(1,2)(3,5,6,4)(7,8)(11,14,15,18)$}{$(3,6,4)(8,9,10)(\bulA{11},12,13)(\bulA{15},16,17)$.}
    \item \head{$(\underline{2\cup4})\cup(\plus{6}\cup\plus{6})$:}
      \triple{$(1,2)(3,4,5,6)(7,8,9,10,\bulA{11},\bulA{12})(\bulB{13},14,15,\bulB{16},17,18)$}{$(1,2)(3,5,6,4)(7,10,11,12)(13,16)$}{$(3,6,4)(7,8,9)(\bulA{11})(\bulA{12})(\bulB{13},14,15)(\bulB{16},17,18)$.}
    \item \head{$(\underline{2\cup4})\cup(\plus{6}\cup\plus{8})$:}
      \triple{$(1,2)(3,4,5,6)(\bulA{7},8,9,\bulA{10},11,12)(\bulB{13},14,15,16,\bulB{17},18,19,20)$}{$(1,2)(3,4,5,6)(7,10)(13,16,17,20)$}{$(\bulA{7},8,9)(\bulA{10},11,12)(\bulB{13},14,15)(\bulB{17},18,19)$.}
    \item \head{$(\underline{2\cup4})\cup(\plus{8}\cup\plus{8})$:}
      \triple{$(1,2)(3,4,5,6)(7,8,\bulA{9},10,11,\bulA{12},13,14)(\bulB{15},16,17,18,\bulB{19},20,21,22)$}{$(1,2)(3,5,6,4)(7,8)(\bulA{9},10,11)(\bulA{12},13,14)(15,18,19,22)$}{$(3,6,4)(8,11,14)(\bulB{15},16,17)(\bulB{19},20,21)$.}
    \end{itemize}

  \item[(4)] {\bf Decomposition of $\uA\cup\uA\cup\uA$ (20 cases)}

    \sethead{$\underline{3}\cup\underline{3}\cup(\underline{2\cup 2})$:}

    \begin{itemize}
    \item \head{$\underline{3}\cup\underline{3}\cup\underline{3}$:}
      \triple{$(1,2,3)(4,5,6)(7,8,9)$}{$(1,3,2)(4,5,6)$}{$(1,3,2)(7,8,9)$.}
    \item \head{$\underline{3}\cup\underline{3}\cup\underline{5}$:}
      \triple{$(1,2,3)(4,5,6)(7,8,9,10,11)$}{$(1,2,3)(7,8,9)$}{$(4,5,6)(9,10,11)$.}
    \item \head{$\underline{3}\cup\underline{3}\cup(\underline{2\cup2})$:}
      \triple{$(1,2,3)(4,5,6)(7,8)(9,10)$}{$(1,2)(4,5)(7,10)(8,9)$}{$(2,3)(5,6)(7,9)(8,10)$.}
    \item \head{$\underline{3}\cup\underline{3}\cup(\underline{2\cup4})$:}
      \triple{$(1,2,3)(4,5,6)(7,8)(9,10,11,12)$}{$(1,2)(4,5)(7,8)(9,11)$}{$(2,3)(5,6)(9,10)(11,12)$.}
    \item \head{$\underline{3}\cup\underline{5}\cup\underline{5}$:}
      \triple{$(1,2,3)(4,5,6,7,8)(9,10,11,12,13)$}{$(4,5,6,8,7)(9,10,11,12,13)$}{$(1,2,3)(6,8,7)$.}
    \item \head{$\underline{3}\cup\underline{5}\cup(\underline{2\cup2})$:}
      \triple{$(1,2,3)(4,5,6,7,8)(9,10)(11,12)$}{$(1,2)(4,5)(6,8)(9,10)$}{$(2,3)(5,8)(6,7)(11,12)$.}
    \item \head{$\underline{3}\cup\underline{5}\cup(\underline{2\cup4})$:}
      \triple{$(1,2,3)(4,5,6,7,8)(9,10)(11,12,13,14)$}{$(1,5,4)(11,12,13)$}{$(1,2,3,4)(5,6,7,8)(9,10)(13,14)$.}
    \item \head{$\underline{3}\cup(\underline{2\cup2})\cup(\underline{2\cup2})$:}
      \triple{$(1,2,3)(4,5)(6,7)(8,9)(10,11)$}{$(1,2)(4,5)(8,10)(9,11)$}{$(2,3)(6,7)(8,11)(9,10)$.}
    \item \head{$\underline{3}\cup(\underline{2\cup2})\cup(\underline{2\cup4})$:}
      \triple{$(1,2,3)(4,5)(6,7)(8,9)(10,11,12,13)$}{$(1,2)(4,5)(8,9)(10,12)$}{$(2,3)(6,7)(10,11)(12,13)$.}
    \item \head{$\underline{3}\cup(\underline{2\cup4})\cup(\underline{2\cup4})$:}
      \triple{$(1,2,3)(4,5)(6,7,8,9)(10,11)(12,13,14,15)$}{$(1,2,3)(4,10)(5,11)(6,7)(8,9)(12,13,14)$}{$(4,11)(5,10)(7,9)(14,15)$.}
    \item \head{$\underline{5}\cup\underline{5}\cup\underline{5}$:}
      \triple{$(1,2,3,4,5)(6,7,8,9,10)(11,12,13,14,15)$}{$(1,4,2,5,3)(6,7,8,9,10)$}{$(1,4,2,5,3)(11,12,13,14,15)$.}
    \item \head{$\underline{5}\cup\underline{5}\cup(\underline{2\cup2})$:}
      \triple{$(1,2,3,4,5)(6,7,8,9,10)(11,12)(13,14)$}{$(1,2,3,4)(6,7,8,9)(11,13)(12,14)$}{$(4,5)(9,10)(11,14)(12,13)$.}
    \item \head{$\underline{5}\cup\underline{5}\cup(\underline{2\cup4})$:}
      \triple{$(1,2,3,4,5)(6,7,8,9,10)(11,12)(13,14,15,16)$}{$(1,2,3,4)(6,7,8,9)(11,12)(13,15)$}{$(4,5)(9,10)(13,14)(15,16)$.}
    \item \head{$\underline{5}\cup(\underline{2\cup2})\cup(\underline{2\cup2})$:}
      \triple{$(1,2,3,4,5)(6,7)(8,9)(10,11)(12,13)$}{$(1,2)(3,5)(6,7)(10,11)$}{$(2,5)(3,4)(8,9)(12,13)$.}
    \item \head{$\underline{5}\cup(\underline{2\cup2})\cup(\underline{2\cup4})$:}
      \triple{$(1,2,3,4,5)(6,7)(8,9)(10,11)(12,13,14,15)$}{$(1,3,5)(11,12,15)$}{$(1,2)(3,4)(6,7)(8,9)(10,15,11)(12,13,14)$.}
    \item \head{$\underline{5}\cup(\underline{2\cup4})\cup(\underline{2\cup4})$:}
      \triple{$(1,2,3,4,5)(6,7)(8,9,10,11)(12,13)(14,15,16,17)$}{$(1,2)(3,5)(8,9,10,11)(14,15,16,17)$}{$(2,5)(3,4)(6,7)(12,13)$.}
    \item \head{$(\underline{2\cup2})\cup(\underline{2\cup2})\cup(\underline{2\cup2})$:}
      \triple{$(1,2)(3,4)(5,6)(7,8)(9,10)(11,12)$}{$(1,3)(2,4)(5,6)(7,8)$}{$(1,4)(2,3)(9,10)(11,12)$.}
    \item \head{$(\underline{2\cup2})\cup(\underline{2\cup2})\cup(\underline{2\cup4})$:}
      \triple{$(1,2)(3,4)(5,6)(7,8)(9,10)(11,12,13,14)$}{$(1,2)(3,4)(5,6)(11,13)$}{$(7,8)(9,10)(11,12)(13,14)$.}
    \item \head{$(\underline{2\cup2})\cup(\underline{2\cup4})\cup(\underline{2\cup4})$:}
      \triple{$(1,2)(3,4)(5,6)(7,8,9,10)(11,12)(13,14,15,16)$}{$(1,3)(2,4)(5,6)(7,8)(9,10)(11,12)(13,14)(15,16)$}{$(1,4)(2,3)(8,10)(14,16)$.}
    \item \head{$(\underline{2\cup4})\cup(\underline{2\cup4})\cup(\underline{2\cup4})$:}
      \triple{$(1,2)(3,4,5,6)(7,8)(9,10,11,12)(13,14)(15,16,17,18)$}{$(3,4,5,6)(9,10,11,12)(15,16)(17,18)$}{$(1,2)(7,8)(13,14)(16,18)$.}
    \end{itemize}
  \end{enumerate}
\end{proof}

\begin{proof}[Proof of Lemma~\ref{lem-supplement}]
  Proposition~\ref{pro-decomposable} implies that every even
  permutation on a set of $7$ or more elements can be written as a
  product of two evenly balanced permutations; moreover, on a set of
  $6$ elements, the only exception is the special atom
  $\underline{2\cup 4}$. This is exactly the content of
  Lemma~\ref{lem-supplement}.
\end{proof}

\section{Conclusion and further work}

We have shown that every even permutation of $2\times A\times 2$ has
alternation depth 9 and even alternation depth 13. The bounds of 9 and
13 are probably not tight. The constructions of
Sections~\ref{sec-first} and {\ref{sec-second}} have many degrees of
freedom, making it plausible that a tighter bound on alternation depth
can be found. The best lower bound for alternation depth known to the
author is 5. An exhaustive search shows that for $A=\s{a,b,c}$, a
3-cycle with support $\s{0}\times A\times \s{0}$ cannot be written
with alternation depth 4. Of course, this particular permutation can
be realized with alternation depth 5 by
Proposition~\ref{pro-construction-1} (and in fact, its even
alternation depth is also 5).

It is reasonable to conjecture that there is nothing special about the
number $2$ in Theorems~\ref{thm-main} and {\ref{thm-even-main}}.
Specifically, if $N$ and $M$ are finite sets, one may conjecture that
there exists a finite bound on the alternation depth of all
permutations $\sigma\in\S(N\times A\times M)$ (or all even
permutations, when $N$ and $M$ are even), for large enough $A$,
independently of the size of $A$.

\section{Acknowledgments}

This work was supported by the Natural Sciences and Engineering
Research Council of Canada (NSERC) and by the Air Force Office of
Scientific Research, Air Force Material Command, USAF under Award
No. FA9550-15-1-0331. This work was done in part while the author
was visiting the Simons Institute for the Theory of Computing.

\bibliographystyle{abbrv}
\footnotesize
\bibliography{alternation}

\end{document}